\documentclass[a4paper, 12pt]{article}
\usepackage{setspace}
\usepackage{etoolbox}

\AtBeginEnvironment{proof}{\doublespacing}

\usepackage[a4paper, total={6in, 8in}]{geometry}
\usepackage{amssymb, stmaryrd}
\usepackage[parfill]{parskip}
\usepackage{amsmath}
\usepackage{bbm}
\usepackage{graphicx}
\usepackage{tikz}
\usetikzlibrary{positioning}
\usetikzlibrary{calc,intersections}
\usepackage{bbm}
\usepackage{pgf}
\usepackage{natbib}
\usepackage{hyperref}

\usetikzlibrary{patterns}
 
\newtheorem{theorem}{Theorem}

\newtheorem{axiom}{Axiom}

\newtheorem{claim}{Claim}

\newtheorem{definition}{Definition}
\newtheorem{example}{Example}

\newtheorem{proposition}{Proposition}
\newtheorem{remark}{Remark}

\newenvironment{proof}[1][Proof]{\noindent\textbf{#1.} }{\ \rule{0.5em}{0.5em}}

\title{Stochastic Choice, Limited Attention, and Aggregation over Attributes}
\author{Arkarup Basu Mallik \thanks{Economic Research Unit, ISI Kolkata. \\
 Email:  arkarupbasumallik97@gmail.com }
 \and Mihir Bhattacharya \thanks{Department of Economics, Ashoka University, Sonipat, Haryana, India. \\ 
Corresponding author email: mihir.bhattacharya@ashoka.edu.in}\and Anuj Bhowmik \thanks{Economic Research Unit, ISI Kolkata. \\ Email: anujbhowmik09@gmail.com} }
\begin{document}
\maketitle                                     
      
\begin{abstract}

\noindent We study an attribute-based model of stochastic choice in which attention to each attribute is limited, and characterize two choice rules within it. The Multiplicative Attention Rule (MAR) rewards an alternative only for the attributes on which it ranks first, combining attention to these multiplicatively across attributes; it captures choice deferral driven by the cognitive load of attending to `too many' attributes. The Additive Attention Rule (AAR) instead averages attention across attributes and credits an alternative by its full ordinal rank within the menu. This difference in aggregation is what separates the two rules behaviourally: AAR reproduces the Attraction Effect while holding the relevant attributes fixed across menus, something MAR can do only when the attribute set itself changes, though AAR, in turn, cannot accommodate the Compromise Effect. Both rules are characterized by axioms on observable choice data, and in both cases the underlying attention parameters and attribute-based preferences are uniquely identified from that data.


\end{abstract}

\newpage

\setcounter{tocdepth}{2}
\tableofcontents

\section{Introduction}

Attributes can provide decision makers with useful rationales for evaluating alternatives (\cite{manzini2007sequentially}). The importance assigned to different attributes, however, may vary across individuals and decisions, and decision makers may not devote equal attention to all of them. For instance, someone choosing a car may place greater emphasis on its performance than on the number of airbags it offers. Such differences in attention can shape how alternatives are evaluated and chosen. In this paper, we characterize \emph{two} complementary choice rules that capture the role of attribute-specific attention in decision making. The first, the Multiplicative Attention Rule (MAR), aggregates attention \emph{multiplicatively} across attributes and is tops-only within each attribute: an alternative receives a reward only when it ranks first on that attribute. The second, the Additive Attention Rule (AAR), aggregates attention through \emph{averaging} and evaluates alternatives according to their full attribute-wise rankings, rather than only whether they rank first. These two rules offer distinct perspectives on how attention shapes choice, allowing us to study when limited attention generates different patterns of behaviour and what can be learned about the decision maker's underlying preferences from observed choices.

The two rules share a common foundation: attention probabilities over attributes and attribute-specific preference orderings over alternatives. Their distinct methods of aggregating attention, however, give rise to different behavioural implications. MAR accommodates choice overload and a high probability of choosing the default (status quo) alternative, whereas AAR reproduces the Attraction Effect without requiring changes in the set of relevant attributes across menus. We provide complete axiomatic characterizations of both rules in terms of observable choice data and establish the unique identification of their attention parameters and attribute-based preferences. Together, these results connect the structure of attention aggregation to observable patterns of choice and clarify the distinct behavioural roles that the two rules can capture.

We provide some examples from consumer behaviour where attention paid to various
attributes plays an important role:

\begin{enumerate}

\item Consider a consumer who wants to buy a smartphone. She may pay more
attention to buying a phone based on screen size instead of the number of
megapixels in its camera. These two attributes are independent so the attention
paid to them is assumed to be independent. We assume that the DM is generally
attentive, i.e., the probability that attention is paid to any attribute is greater
than half.\footnote{This is only needed for MAR; AAR is more general in this regard.} In other words, each relevant attribute is more likely to be considered
than not.

\item Consider a consumer who wants to buy a box of breakfast cereal. She may be
more focused on nutritional attributes and less on the type of fruit that comes
with it. The difference in attention probabilities to the two attributes may affect
the choice probability of a specific cereal box. However, unlike our first premise,
in this case, the attributes are not independent (since the nutritional attributes
and the type of fruit associated with the cereal are correlated). In our model,
we assume that the attention paid to the attributes is independent.\footnote{\cite{bialkova2010determines} identify the factors affecting the attention paid to different
nutritional labels in an experimental setting.}
\end{enumerate}
We take choice probabilities to be observable for any given subset of alternatives and any subset of attributes (in fact, the parameters of our model can be identified using only choice probabilities from binary comparisons of alternatives under all singleton attribute sets (for the preference orderings) and from singleton alternative sets under all singleton attribute sets (for the attention parameters)). This is a natural assumption: in most online and offline markets, a planner can observe which alternatives are on offer and which attributes are associated with them. Online retailers routinely let shoppers filter alternatives by attribute ratings, product type, dimensions, delivery time, so that choice data broken down by attribute set is often directly available. Our rules are designed to capture how attention that shifts across attributes shows up in exactly this kind of data.

We first characterize Multiplicative Attention Rule
(MAR) according to which the probability of an alternative $x$ is the product of
the probabilities with which the DM pays attention to the attributes under which $x$
is ranked highest and the weighted probability with which attention is not drawn by
those attributes under which $x$ is not the highest ranked. The ranking over alternatives is a linear ordering for each attribute. First, we provide an example of the rule
with three alternatives and two attributes.

\begin{example}\label{egA}
Suppose $S = \{x, y, z\}$ and $A = \{a_1, a_2\}$. Let the two orderings with
respect to the attributes $a_1$ and $a_2$ be $x \succ_{a_1} y \succ_{a_1} z$ and $y \succ_{a_2} x \succ_{a_2} z$ respectively.
Let $\gamma(a_1) = 0.6$ and $\gamma(a_2) = 0.9$ be the two attention probabilities of the attributes.\footnote{Note that the attention paid to attributes need not sum to one since they are assumed to be
independent.}
The choice probabilities of the three alternatives from the set $S = \{x, y, z\}$ according
to MAR, $p^S$, are as follows:
\begin{align*}
\text{(i)}\quad & p^S(x, S, A) = \gamma(a_1)\left( \frac{1-\gamma(a_2)}{|S|-1}\right)  = 0.6(0.05) = 0.03. \\[4pt]
\text{(ii)}\quad & p^S(y, S, A) = \gamma(a_2)\left( \frac{1-\gamma(a_1)}{|S|-1}\right)  = 0.9(0.2) = 0.18. \\[4pt]
\text{(iii)}\quad & p^S(z, S, A) = \left(\frac{1-\gamma(a_2)}{|S|-1}\right)\left( \frac{1-\gamma(a_1)}{|S|-1}\right) = \left( \frac{1-0.9}{3-1}\right)\left( \frac{1-0.6}{3-1}\right) = 0.01.
\end{align*}
\end{example}

Note that the choice probability of alternative $y$ is the highest (6 times that of $x$
and 18 times that of $z$) since the attention paid to the attribute ($a_2$) in which it is
ranked first is much higher than the attention paid to the attribute ($a_1$) where it is
not (0.9 as compared to 0.6). Moreover, the probability of ignoring the attribute $a_1$
in which it is not ranked first is also relatively high (0.4). Also note that the choice
probability of the default alternative is the residual probability given by $p^S(x^*, S, A) =
1 - 0.03 - 0.18 - 0.01 = 0.78$.

An observation from the above example is the high probability of the default alternative.\footnote{We discuss the high probability of the default alternative in detail in Appendix \ref{sec7}.}
An interpretation of this property in the example is the fact that DM pays
attention to both attributes which have different top-preferred alternatives. This
`confuses'\footnote{For a small number of attributes, this `confusion' is necessary and sufficient for the default
alternative to be chosen with the highest probability. However, for a large number of attributes,
even in absence of confusion, the default alternative can be chosen with the highest probability. This
is because the large number of attributes itself exerts high cognitive load on the DM and she resorts
to her default alternative more often than not. We get back to this point in Appendix \ref{sec7}.} the DM since she does not choose according to the relative positions of
$x$ and $y$, which are similar across the two attributes. Even though she pays higher
attention to the second attribute, the choice probability of $y$ is multiplied by the
probability of not paying attention to the first attribute (weighted by $\frac{1}{2}$). This discounting drags down the choice probability of alternative $y$ in the set, resulting
in the DM choosing the default with a higher probability.

MAR can explain cognitive constraints while choosing when there are numerous
attributes under consideration even if the DM is highly attentive. Consider the following situation: If the number of attributes is too high and even if there is a unique
alternative that is ranked highest in all of them, then it is plausible that the probability with which that alternative is chosen is not very high.\footnote{In Appendix \ref{sec7}, we also show that MAR is robust in terms of relative choice probabilities
in the presence of a large number of attributes; see Example \ref{egRob} in particular.} This captures the cognitive load caused by having to go through every
attribute, which leads to choice deferral.\footnote{\cite{iyengar2000choice} provide empirical evidence that the presence of more alternatives
can lead to choice deferral. While our paper has an attributes-based explanation, the essential
motivation is similar in the sense that the presence of more entities can lead to more cognitive load.
In our case, the presence of more attributes can lead to a choice deferral.} This also implies a high choice probability of
the default alternative, since the DM has too many attributes to focus on. Also, DMs may choose to window shop (the act of just browsing various
products) more often than she actually buys a product.\footnote{Another similar scenario related to this is the DM wanting to browse products online while
procrastinating.}

We use seven axioms to characterize MAR. Two carry most of the weight. Independence (IN) is a stochastic analogue of Plott's Path Independence axiom (\cite{plott}): choice probabilities under a union of disjoint attribute sets factor as the product of the choice probabilities under each set separately, holding the menu fixed. Uniformity of Inferior Alternatives (UIA) requires that, under a single attribute, any two alternatives beaten pairwise by a common alternative $x$ receive equal choice probability from the full menu; here, only the top-ranked alternative is ever distinguished. The remaining five axioms are more standard. Full Support (FS) requires choice probabilities under a single attribute to exhaust the menu, so the default is never chosen when there is more than one alternative to compare. Invariance of Singletons (IS) requires that a singleton menu's choice probability, under a singleton attribute set, not depend on which alternative is offered. Dominance (DOM) is an attribute-wise analogue of independence of irrelevant alternatives: removing an alternative that loses every pairwise comparison to $x$ under a given attribute leaves $x$'s choice probability under that attribute unchanged. Binary Difference (BD) rules out ties in binary comparisons under a single attribute, and Stochastic Transitivity (ST) is the natural probabilistic analogue of transitivity, both restricted to singleton attribute sets.

Our first main result (Theorem \ref{thm1}) shows that these seven axioms characterize MAR. Their value lies less in any single axiom than in what the two doing the real work jointly pin down: IN forces attention to aggregate \emph{multiplicatively} across attributes, and UIA allows within-attribute credit to collapse onto the top-ranked alternative alone, discarding all information about how the remaining alternatives are ranked. Both features are derived from behaviour that is fully observable in binary comparisons under one attribute at a time. This is also what sets up the paper's central comparison: in Section \ref{secAAR}, we replace IN and UIA with axioms that account for an alternative's full rank rather than only whether it ranks first. We show that this change in the axioms, while keeping the underlying primitives unchanged, leads to the Additive Attention Rule.

We show in Section \ref{sec4.1} that all the parameters in MAR are uniquely identifiable
from the choice data. BD plays an important
role in this by ensuring the difference in the choice probabilities from binary sets. This
leads to the attention parameters being distinct for a given choice data since no two
data sets would generate the same attention parameters and preference orderings. We next show that MAR satisfies a well-known property Regularity
with respect to alternatives (for a fixed set of attributes), and also satisfies regularity
with respect to attributes (for a fixed set of alternatives). The former requires that
the choice probability of an alternative from a subset of alternatives $S \subseteq T$ be weakly
greater than from $T$ when the set of attributes is fixed. The latter requires that the
choice probability of an alternative from a subset of attributes $A \subseteq B$ be weakly
greater than that from $B$ when the set of alternatives is fixed. This is the standard
regularity axiom defined over alternatives as in \cite{fudenberg2015stochastic}, \cite{mcfadden1990stochastic} and \cite{aguiar2019adaptive}. A violation of regularity can also
be captured when the set of attributes and the set of alternatives are allowed to vary.
A justification for this is the Attraction Effect as introduced in \cite{huber1982adding} 
which causes a reversal in the choice probabilities.


The regularity of MAR with respect to alternatives comes with an important restriction. With multiplicative aggregation across attributes and a tops-only evaluation within each attribute, removing alternatives cannot reduce the choice probability of any remaining alternative when the attribute set is fixed. Thus, MAR cannot generate the Attraction Effect (\cite{huber1982adding}) on a fixed attribute set. Within our framework, the effect can arise under MAR only if the addition of the asymmetrically dominated decoy also changes the relevant attribute set. For example, in Section \ref{sec4.2}, the relevant attributes shift from $\{a,b\}$ to $\{a,c\}$. Hence, under MAR, the Attraction Effect requires a change in the attributes considered relevant to the choice problem.

This motivates our second main contribution. We introduce the Additive Attention Rule (AAR), which uses the same primitives as MAR, namely attribute-specific attention parameters and attribute-specific linear orders over alternatives, but aggregates them differently. Attention is averaged across attributes, and an alternative is weighted according to its full ordinal rank within each attribute rather than only whether it is top-ranked. This change in aggregation allows AAR to generate the Attraction Effect while keeping the set of attributes fixed (Section \ref{secATTAAR}). In particular, adding an asymmetrically dominated alternative can reverse the choice ranking between two existing alternatives even though the attributes relevant to the choice remain unchanged. 

Under AAR, the same $S$, $A$, and attention weights as Example \ref{egA} give choice probabilities $0.3$, $0.325$, and $0.125$ for $x$, $y$, and $z$, against MAR's $0.03$, $0.18$, and $0.01$ (Example \ref{egAdditive}): crediting an alternative for its full rank, rather than only for being on top, keeps $x$ and even $z$ in contention and more than halves the default's share.

We provide a characterization of AAR (Theorem \ref{thmAAR}) by replacing IN, FS, UIA, and DOM with three axioms suited to additive aggregation: Attribute Average (AVG), Menu Additivity (MA), and Binary Calibration (BC), while retaining IS, BD, and ST. The same linearity that allows AAR to generate the Attraction Effect, however, rules out the Compromise Effect (\cite{simonson1989choice}): a rank-two alternative cannot simultaneously receive greater weight than two extreme alternatives (Section \ref{secAARcomp}). These results highlight the distinct implications of multiplicative and additive attention aggregation and the behavioural phenomena each can accommodate.

Many papers have studied the salience of attributes in individual stochastic choice.
\cite{bordalo2012salience} models the salience of differences in expected payoffs in a model
of risk, \cite{bordalo2013salience} and \cite{bordalo2020memory} model the role of memory in
decision-making which explains many behavioural biases including inattention to hidden product attributes. \cite{manzini2012categorize} shows that a choice rule which
takes into account different categories before choosing an alternative satisfies weak
WARP in a deterministic setting. The choice rules characterized in our setting can violate weak WARP and can even violate independence of irrelevant alternatives.

We contribute to the literature on Multiple Attributes (or Criteria) based Decision
Making, Limited Attention, and Stochastic Choice. \cite{lancaster1966new} introduced a
model with multiple attributes being associated with each alternative. He explicitly
considers alternatives as a vector of attributes and provides a theory on the optimization of the DM subject to a budget constraint on the attributes space. There are
papers which study the case where the criteria can be interpreted as multiple selves
or rationales, e.g. \cite{kalai2002rationalizing}, \cite{manzini2007sequentially}, \cite{cherepanov2013rationalization}, and \cite{hara2019coalitional}. In the limited attention literature, \cite{masatlioglu2012revealed} provides conditions for revealed preference, while \cite{lleras2017more} provides
a choice-theoretic foundation of limited attention and consideration sets. A good
starting point for the literature on stochastic choice is \cite{luce}.

\cite{manzini2014stochastic} is closely related to our paper because of the similarity in the formulation of the stochastic choice rule, particularly MAR. Their model, however, captures limited attention over alternatives, with attention parameters defined at the level of alternatives. In our model, attention is instead defined over attributes, and both the set of alternatives and the set of attributes may vary across choice problems. 

Another closely related paper is \cite{bhattacharya2021frame}, which studies stochastic choice through frames. Although their framework also involves limited attention, their choice rule differs from ours. In particular, MAR and AAR partition the attribute set differently across alternatives, which leads to different functional forms for the choice probabilities and a different distribution of probability across dominated alternatives. Moreover, in \cite{bhattacharya2021frame}, frames do not provide additional information for the rational evaluation of alternatives, as in \cite{salant2008f}. In our model, by contrast, attributes are informative and enter directly into the DM's evaluation of alternatives. 

Our approach requires a richer observable choice data set than these related models: choice probabilities must be observed for every set of alternatives and every set of attributes. This informativeness does not require a substantial additional data burden for identification. Although $p$ is defined over all sets of alternatives and attributes, the parameters of both MAR and AAR are identified from choice probabilities for binary comparisons under singleton attribute sets and for singleton alternative sets under singleton attribute sets. Thus, identification requires only a relatively small, attribute-specific slice of the observable choice data.

\cite{gul2014random} uses attributes in a stochastic choice setting. They consider a setup where the attributes are not explicitly modeled and are revealed from the stochastic
choice data. Their axioms and the structure of the choice rule are very different as
compared to ours. They seek to explain the problem of duplicates infested in the
Luce rule (\cite{luce}). An interesting model is provided by \cite{honda2021categorical} where
stochastic choices are made when each alternative can be associated with a category.
The stochastic choice rule that is used has two terms (in a product form). The first
term is representative of the probability of the categories in which there are better
alternatives (than the said alternative) not being considered. The second term is the
probability of consideration of the category containing the said alternative multiplied
by its Luce weight with respect to the other alternatives in the same category. The
weights in MAR and AAR just depend on the cardinality of the menu
(of alternatives) and the cardinality of the attribute set.

\cite{cattaneo2020random} provides a model where attention is both limited and random.
They have an axiom that states that the probability of an alternative being considered
weakly decreases as the menu of alternatives expands. MAR, on the other hand, satisfies a different but related form of monotonicity: Regularity in Attributes. The
situation is similar: If no new alternatives are added but the relevant attribute set
expands, attention diminishes. Hence, each alternative is less
likely to be considered (according to their model) and its choice probability also
decreases. AAR, by contrast, can accommodate some violations of regularity.

\cite{kovach2022focal} is another stochastic choice model with limited attention where some alternatives have more attention associated with them, known as the Focal Alternatives. \cite{kovach2023reference} presents a model
where reference points direct attention and shows how stochastic choices are made
when the attention is reference-dependent and random. \cite{ahumada2018luce}
explains a model where a DM uses the Luce Rule on a consideration set. \cite{kimya2018choice} gives a detailed model of consideration set formation in the presence of attributes and limited attention. However, his model is based on a deterministic choice
setup with thresholds, unlike ours. Also, for various forms of Stochastic Transitivity,
see \cite{Fishburn1973}.

In addition to the above, there are other papers that capture the role of attributes on
stochastic choice. \cite{galperti2024listing} considers the case where people's decisions depend
on the \textit{order} in which the attributes of the available options are presented to them.
They cite sufficient empirical evidence which confirms that such framing effects exist.
In particular, \cite{bergus1995clinical}, \cite{cunnington1997effect}, and \cite{chapman2000cognitive} show that doctor's diagnoses (which can be interpreted as the chosen
alternative) can depend on the order of attributes (symptoms) in which they are
presented. It is plausible that the \textit{set} of attributes presented to the DM while looking
at a menu of alternatives can influence choice. We provide a theory of stochastic choice
motivated by this. This is relevant for choice architecture (through advertising on
particular attributes), where firms focus on directing the DM's attention to
those attributes so as to influence choice. For example, shopping websites such as
Amazon and Best Buy often emphasize certain attributes of their products that draw
additional attention to them.

The paper is divided into sections. Section \ref{sec2} describes the model, while Section \ref{sec3} introduces MAR and provides the axioms required for the main result on MAR. Section \ref{sec4} provides the
results and properties satisfied by MAR. Section \ref{secAAR} develops the Additive Attention
Rule in full: the rule itself, the mechanism behind the Attraction Effect under a fixed
attribute set, its axiomatic characterization, the impossibility of the
Compromise Effect, and a discussion on why a simpler version of the rule fails. Section \ref{sec5} concludes by drawing together what the multiplicative
and additive rules each contribute. Appendix \ref{sec7} provides a detailed discussion of the
link between MAR and Pareto dominance. In particular, it provides more motivational examples, explains when the default alternative is chosen more often than not, and shows that under full attention a deterministic choice function based on MAR is Pareto dominant. Appendix \ref{sec6} establishes the independence of the axioms for both MAR and AAR.

\section{Model}\label{sec2}

\par Let $\mathbf{X}$ be the set of all alternatives and \(\mathcal{X}\) be the set of all finite non-empty subsets of \(\mathbf{X}\). Let \(\mathbf A\) be the set of all attributes and \(\mathcal{A}\) be the set of all finite non-empty subsets of \(\mathbf{A}\). We define stochastic choice rules which provide the choice probability of an alternative for every $S\in \mathcal{X}$ and $A\in \mathcal{A}$.

\begin{definition}[Stochastic Choice Rule (SCR)] An SCR, $p:\mathbf X\times \mathcal{X}\times \mathcal{A}\to [0,1]$ for any $x\in \mathbf{X}, S\in\mathcal{X}$ and $A\in\mathcal{A}$ provides 
\(p(x,S,A)\) which is  the probability that alternative \(x\in \mathbf X\) is chosen from a menu \(S\in \mathcal{X}\) and the set of relevant attributes \(A \in \mathcal{A}\). We denote $x^*\notin \mathbf{X}$ as a default alternative which is an outside option or the act of not choosing any alternative from $S$ and \(p(x^{*}, S, A)=1-\sum_{x \in S}p(x,S,A).\) 
\end{definition}

The choice probabilities over the set of alternatives and the default alternative sum to one, i.e., \(\sum_{x \in S}p(x,S,A)+p(x^{*}, S, A)=1\), where \(x \in \mathbf X\), \(S \in \mathcal{X}\), and \(A \in \mathcal{A}\). Note that \(p(x,S,A)=0\) if \(x \not \in S\). 

\par We introduce some notations. Suppose that $\succeq_{a}$ is a linear ordering over $\mathbf{X}$ according to the attribute $a\in A$.\footnote{We say that $R$ is a linear ordering if it is (i) complete: $xRy$ or $yRx$ for all $x,y$, (ii) antisymmetric: $xRy$ and $yRx$ implies that $x=y$, and (iii) transitive: $xRy$ and $yRz$ implies $xRz$ for any $x,y,z$.} This is also the preference of the DM on the given attribute $a\in \mathbf A$. We say that an alternative \(x\) \textit{dominates} another alternative \(y\) via an attribute \(a\) if \(x \succeq_{a} y\). For any $A\in \mathcal A$, $S\in \mathcal X$ and  $x\in S$,
 the set 
\[
A_{x}(S)=\{a \in A: x \succeq_{a} y,  \forall y \in S\}
\]
 denotes the attributes in which $x$ dominates all the alternatives in $S$ and let 
\[
A_{x}^c(S)=\{a \in A: \exists y \in S \backslash \{x\}, y \succeq_{a} x\}
\] 
be the set of attributes where $x$ is dominated by some alternative $y\in S\backslash\{x\}$. 

 \section{The Multiplicative Attention Rule}\label{sec3}
 We are now in a position to define our first stochastic choice rule: \textit{Multiplicative Attention Rule $($MAR$)$}. 

\begin{definition}[MAR] An SCR, $p^{S}$ is a MAR if for all $A\in\mathcal{A}$ there exists a linear ordering $\succeq_{a}$ on $\mathbf X$ and attention parameters $\gamma(a)\in (\dfrac{1}{2},1]$ for all $a\in \mathbf A$ such that for any $S\in\mathcal{X}$ and $x\in S$,   \begin{displaymath}
p^{S}(x,S,A)=\begin{cases}
\prod_{a\in A_{x}(S)}\gamma(a)\dfrac{\prod_{b\in A_{x}^{c}(S)}(1-\gamma(b))}{(|S|-1)^{|A_{x}^{c}(S)|}}, &\quad \text{ if }  A_{x}(S) \not = \emptyset,  A_{x}^{c}(S) \not = \emptyset; \\
\ \prod_{a \in A_{x}(S)}\gamma(a), & \quad \text{if } A_{x}(S)=A; \\
 \dfrac{\prod_{b \in  A_{x}^{c}(S) }(1-\gamma(b))}{(|S|-1)^{|A_{x}^{c}(S)|}}, &\quad \text{ if } A_{x}^{c}(S)=A. \\   
    \end{cases}
    \end{displaymath}
\end{definition}

 According to MAR, the probability of an alternative $x$ from a set of alternatives $S$ and given a set of attributes $A$ is the product of the probability of paying attention to the set of attributes in which $x$ ranks first and the probability with which attention is not paid to other attributes in which $x$ is not ranked first where the latter is weighted by $\dfrac{1}{|S|-1}$. Therefore, the choice probability is weighted down by the attributes in which it does not rank first and the number of alternatives in $S$. Note that \(\gamma(a)\) is the probability that the relevant attribute \(a\) is considered or the attention probability of attribute $a$. The rule also considers high attention to the attributes, i.e., $\gamma(a) > \dfrac{1}{2}$, for all $a \in A$ where \(A\) is the relevant set of attributes. This assumption means that each \textit{relevant} attribute is more likely to be considered than not. Another implication of this assumption is that if there is only one attribute then it selects the highest-ranked alternative in that attribute with a greater probability than the default alternative. 

We introduce the Axioms required to characterize MAR:
\begin{axiom}[Independence (\textbf{IN})] An SCR $p$ satisfies \textbf{IN} if for any $S\in\mathcal{X}$ and $A, B \in \mathcal{A}$ such that $A\cap B=\emptyset$, we have 
        \(p(x,S, A \cup B)=p(x,S,A).p(x,S,B)\).
    \end{axiom}
    
    \textbf{IN} states that whenever two relevant attribute sets are disjoint, the choice probability for the union of the two sets is equal to the product of the choice probabilities given the two sets separately (keeping the set of alternatives fixed). \textbf{IN} is related to the \textit{Path Independence} axiom as in \cite{plott} which states that $\forall S,T \in \mathcal{X}, S\cap T= \emptyset, c(S\cup T)=c(c(S)\cup c(T))$\footnote{$c$ is a choice correspondence, which maps from a feasible menu of alternatives to the best/acceptable alternatives in that menu. Formally, $c:\mathcal{X}\rightarrow\mathcal{X}$ and $c(S) \subseteq S, \forall S \in \mathcal{X}$.}. \textbf{IN} is a stochastic analogue of this axiom (although not exactly) applied to the attributes set by keeping the alternatives set fixed. Thus, if an alternative gets a high choice probability in $(S,A)$ and if it also gets a high choice probability in $(S,B)$, then it should also get a relatively high choice probability\footnote{Note that the choice probability $p(x,S,A\cup B)$ is somewhat dampened even if the choice probabilities $p(x,S,A)$ and $p(x,S,B)$ are high. This is because the cognitive load of considering attributes in $(S,A \cup B)$ is higher than the cognitive load of considering attributes in either of $(S,A)$ and $(S,B)$.} in $(S, A \cup B)$ and if it gets a low choice probability either in $(S,A)$ or in $(S,B)$, then it should also get a low choice probability in $(S,A \cup B)$. So, \textbf{IN} posits that the choice probabilities are consistent while combining different attribute sets given that the set of alternatives is held fixed.

    The next Axiom applies to cases with a singleton attribute set. 
    
\begin{axiom}[Full Support (\textbf{FS})] An SCR $p$ satisfies \textbf{FS} if for any $S\in\mathcal{X}$, \(|S| \geq 2\) and any $a\in \mathbf A$,   \(\sum_{x \in S}p(x,S,\{a\})=1\).
\end{axiom}

\textbf{FS} states that under a single attribute and any set of alternatives, with at least two alternatives, the choice probabilities over the alternatives excluding the default alternative sum to one.

When there are two or more alternatives and a single attribute, the DM prefers variety in the menu of alternatives and the default alternative $x^{*}$ is not chosen at all. This is because with $|S| \geq 2$ and a singleton attribute, it is easy to compare alternatives in a single dimension and the DM chooses some alternative or the other as there is a freedom of choice or comparison over alternatives and at the same time, it is easier to focus with a single relevant attribute. However, if there is only one alternative and one attribute, then the DM may not choose that alternative with full support and the default alternative has a positive choice probability. This is due to the fact that there is no other alternative to compare with other than the default alternative which enhances the choice probability of the latter. Alternatively, the DM might not always want to choose the last piece of cake (the only alternative available) as this shows a lack of freedom of choice. There is a great deal of literature available on freedom of choice (see, for example, \cite{barbera2004ranking} and \cite{puppe1996axiomatic}).

\begin{axiom}[Invariance of Singletons (\textbf{IS})]  An SCR $p$ satisfies \textbf{IS} if for any $a\in \mathbf A$ and $x,y\in \mathbf{X}$,
    \(p(x, \{x\}, \{a\})=p(y, \{y\}, \{a\}).\)
\end{axiom}

\textbf{IS} states that if there is only one relevant attribute and only one alternative is available, then the choice probability of that alternative is independent of the alternative. In this case, the choice probability of the alternative is driven by the relevant attribute. Alternatively, the probability of choosing the default alternative from singleton alternative sets with singleton relevant attributes is independent of which other singleton alternative is in question. Here, if the DM has a lack of freedom of choice, she does not care which singleton alternative is available to her and she just cares about the attribute in question.\footnote{In fact, MAR satisfies a stronger version of \textbf{IS}: For any $x,y \in \mathbf{X}$ and for any $A \in \mathcal{A}$, $p(x, \{x\}, A)=p(y, \{y\}, A)$. This shows that if only one \textit{alternative} is available, then the choice probability is driven solely by the set of relevant attributes (singleton or not) $A$ in question. However, for the purpose of our characterization, the weaker version, \textbf{IS}, is enough. } As an example, suppose that the alternative in question is a newspaper. If a single newspaper is available in the economy, then the DM may think that the newspaper in question lacks autonomy and only relays the messages of the government. Thus, it is plausible that the DM may want to make her choices based on how she feels about the government (which is analogous to the attribute in question in our case) in power and is invariant to whichever single newspaper is in circulation.\footnote{An alternative explanation based on limited attention is that in our model, choice is random due to imperfect attention and the attention drawn to an alternative is influenced only by attributes. See \cite{bhattacharya2021frame} for a more detailed discussion of this axiom along these lines.}

\begin{axiom}[Uniformity in Inferior Alternatives (\textbf{UIA})] An SCR $p$ satisfies \textbf{UIA} if for all $S\in\mathcal{X}$, $a\in \mathbf A$, 

\[\big[\exists\; x,y,z\in S \text{ s.t. } p(x,\{x,y\},\{a\}) > p(y,\{x,y\}, \{a\}) \]\[\text{ and } p(x,\{x,z\},\{a\}) > p(z,\{x,z\}, \{a\})\big]\]\[ \implies \big[p(y,S,\{a\}) = p(z,S, \{a\})\big].\]
\end{axiom}

This Axiom states that (when a single attribute is considered) if there is an alternative $x$ that has a higher choice probability than two other alternatives \(y,z\) when offered in two menus $\{x,y\}$ and $\{x,z\}$ respectively, then the choice probabilities of the latter two alternatives are equal in the menu $S$ whenever $x,y,z \in S$. Since the choice probability of \(x\) is higher than that of \(y\) and \(z\) in the binary sets $\{x,y\}$ and $\{x,z\}$, respectively, \(y\) and \(z\) are inferior to $x$ according to the choice probability in a binary comparison. This Axiom states that the inferior elements in a menu are uniform in terms of choice probability (where only a singleton attribute is relevant). In some sense, either of the inferior alternatives (which are $y,z$) cannot be the most preferred alternative in the set $S$ that contains the superior alternative $x$. This reflects the fact that only the \textit{top} alternative matters (in the case of singleton attribute sets) and the inferior alternatives are chosen with equal probability. One can ask that, given the premises of \textbf{UIA} ($x$ being superior to both $y$ and $z$), and assuming without loss of generality that $y$ is superior to $z$ (that is, $p(y, \{y,z\},\{a\})> p(z, \{y,z\},\{a\})$), should it not be the case that $p(y,S,\{a\}) > p(z,S,\{a\})$ (instead of  $p(y,S,\{a\}) = p(z,S,\{a\})$ as required by \textbf{UIA})? In this context, we suggest that the answer is negative.\footnote{This is relaxed in AAR.} This is because the DM cares only about the top alternative and additionally, has limited attention. We note that although $y$ is superior to $z$, it is also true that $y$ is not the top alternative in $(S, \{a\})$ as long as $x \in S$. While $y$ and $z$ are chosen with equal probability in $(S, \{a\})$, note that \textbf{UIA} allows for $p(y,T, \{a\})>p(z,T,\{a\})$, where $x \not \in T$ and $y$ is the top alternative in the choice problem, $(T,\{a\})$,\footnote{Such a $T$ always exists; for example, take $T=\{y,z\}$.} in terms of binary choice probabilities (that is, $p(y, \{w,y\}, \{a\})>p(w,\{w,y\},\{a\}), \forall w \in T \setminus\{y\}$). In summary, \textbf{UIA} allows the DM to `experiment' among inferior alternatives (whenever a single attribute is relevant) with equal choice probabilities.

\begin{axiom}[Dominance (\textbf{DOM})] An SCR $p$ satisfies \textbf{DOM} if for all $a\in \mathbf A$ and for all $S\in \mathcal{X}$ such that $x,y\in S$,
 \[ 
\big[p(x, \{x,y\}, \{a\}) > p(y,\{x,y\}, \{a\}) \mbox{ for all } y\in S\setminus \{x\}\big]
\]
\[
\implies \big[p(x, S, \{a\}) = 
p(x,S \backslash \{y\}, \{a\}) \mbox{ for all } y\in S\setminus \{x\}\big].\] 
\end{axiom}

\textbf{DOM} states that for any menu $S$ and any alternative $x$ when only one attribute is considered if the choice probability of \(x\) is strictly greater than that of \(y\) 
when offered in the menu $\{x, y\}$ for any $y\in S\backslash\{x\}$ then the choice probability of \(x\) does not change when $y$ is removed from $S$ with the same attribute being considered. This is because if $x$ was chosen with a higher probability than $y$ in $(\{x,y\}, \{a\})$ then it is conceivable that $x$ is better than $y$ in the attribute $a$ as it is easier to focus with a single relevant attribute. Thus, the removal of any dominated element $y$ should not influence the probability of choosing the dominant element $x$ when the same attribute $a$ is in question. This axiom is also similar to the independence from irrelevant alternatives (IIA), which is widely used in the individual and collective choice literature, which is suitably modified in our framework (in our case, the irrelevant alternatives are the dominated alternatives).

\begin{axiom}[Binary difference (\textbf{BD})] An SCR $p$ satisfies \textbf{BD} if for all $a\in \mathbf A$ and for all \(x,y \in X, x\neq y\),
    \(p(x,\{x,y\},\{a\}) \not = p(y,\{x,y\},\{a\})\).
\end{axiom}

This Axiom rules out equality of choice probabilities when the menu consists of two alternatives and a single attribute. In such situations, the DM chooses one of the two distinct alternatives with a strictly higher probability than the other. This is similar to the strict preference restriction on binary comparison between alternatives, which is widely used in social choice theory and individual choice theory. We use a weak form where the strict inequality of the binary choice probabilities is required to hold only for singleton attribute sets.

\begin{axiom}[Stochastic Transitivity (\textbf{ST})]
An SCR $p$ satisfies \textbf{ST} if for all $a\in \mathbf A$ and for all \(x,y,z \in X\), 
\[\Big[p(x,\{x,y\},\{a\}) > p(y,\{x,y\},\{a\}) \text { and } p(y,\{y,z\},\{a\}) > p(z,\{y,z\},\{a\})\Big] \] \[\implies p(x,\{x,z\},\{a\}) > p(z,\{x,z\},\{a\}).\]
\end{axiom}

This Axiom is a variation of Weak Stochastic Transitivity in \cite{Fishburn1973} with strict inequalities and holds when a single attribute is considered. It is a straightforward version of the transitivity postulate modified for the attribute setting. Although there is some experimental evidence against transitivity, it is used in most choice-theoretical models. However, we weaken the transitivity and apply it only to cases with singleton attribute sets.


\subsection{Results on MAR}\label{sec4}

We are now ready to state our first main result.

\begin{theorem}\label{thm1} An SCR satisfies \textbf{IN}, \textbf{FS}, \textbf{IS}, \textbf{UIA}, \textbf{DOM}, \textbf{BD} and \textbf{ST} if and only if it is a \textbf{MAR}.
\end{theorem}

The proof works from the bottom up. We first use BD and ST on binary comparisons under a single attribute to construct the attribute-specific linear order $\succeq_a$, then use DOM to show that an alternative's choice probability under a single attribute is unaffected by removing anything it beats pairwise, so it only matters whether $x$ is $\succeq_a$-top in $S$, not by how much. FS and UIA then pin down the exact value of that probability, first for the top-ranked alternative and then, by uniformity, for everyone below it. IN does the remaining work: once to split a probability over a union of attribute sets into a product across singletons, and once more to reassemble the pieces into the multiplicative form of Definition \ref{aardef}. 

\begin{proof} (\(\implies\)) Let \(p\) be an SCR that satisfies all the Axioms stated in the statement of Theorem \ref{thm1}. We define a binary relation on \(\mathbf X\) for each attribute \(a \in \mathbf A\) and denote it as \(\succeq_{a}\) as follows:
\begin{displaymath}
    x \succeq_{a} y \iff p(x,\{x,y\},\{a\}) \geq p(y,\{x,y\},\{a\}).
\end{displaymath}

We show that $\succeq_{a}$ is a well-defined linear ordering. We first show that $\succeq_{a}$ is complete. Take any $x,y\in \mathbf X$. If $x=y$, then  $p(x,\{x,y\},\{a\})=p(y,\{x,y\},\{a\})$. Therefore, $[x=y]\implies [x\succeq_{a}y \text{ and } y\succeq_{a} x]$. If $x\neq y$, by \textbf{BD}, $p(x,\{x,y\},\{a\})\neq p(y,\{x,y\},\{a\})$.  Therefore, either $p(x,\{x,y\},\{a\})> p((y,\{x,y\},\{a\})$ which implies $x\succeq_{a}y$ or $p(y,\{x,y\},\{a\})> p(x,\{x,y\},\{a\})$ which implies that $y\succeq_{a} x$. Therefore, $\succeq_{a}$ is complete.

We show that \(\succeq_{a}\) is antisymmetric. Let \(x\succeq_{a}y\) and \(y\succeq_{a}x\) for some $x,y\in \mathbf X$. By definition of $\succeq_{a}$, 
    \[\big[x\succeq_{a}y\big] \implies [p(x,\{x,y\},\{a\}) \geq p(y, \{x,y\}, \{a\})]\]
    \[[y\succeq_{a}x]\implies \big[p(y,\{x,y\},\{a\}) \geq p(x, \{x,y\}, \{a\})\big]\]
    
    Above two statements imply that $p(x, \{x,y\}, \{a\})=p(y, \{x,y\},\{a\})$. However, by \textbf{BD} this can only be true if $x=y$. Hence, \(\succeq_{a}\) is antisymmetric. The 
strict component of $\succeq_{a}$ can be derived as follows: $x\succ_{a}y\iff p(x,\{x,y\}, \{a\}) > p(y,\{x,y\},\{a\})$ for any $x,y\in \mathbf X$ and $a\in A$.

For transitivity, let \(x \succeq_{a}y\) and \(y \succeq_{a} z\) for some $x,y,z\in \mathbf X$. Therefore,  
    $p(x,\{x,y\},\{a\}) \geq p(y,\{x,y\}, \{a\}) $ $\text{ and } p(y,\{y,z\},\{a\}) \geq p(z,\{y,z\}, \{a\}).$ As $x,y,z$ are distinct, by \textbf{BD}, the two aforementioned inequalities are strict. By \textbf{ST},
    
\begin{displaymath}
    \Big[p(x,\{x,y\},\{a\}) > p(y,\{x,y\}, \{a\}) \text{ and } p(y,\{y,z\},\{a\}) > p(z,\{y,z\}, \{a\})\Big]
\end{displaymath} 
\begin{displaymath}
    \implies  \big[p(x,\{x,z\},\{a\}) > p(z,\{x,z\}, \{a\})\big].
\end{displaymath} Hence, \(x \succeq_{a} z\) which establishes that \(\succeq_{a}\) is transitive. Thus, \(\succeq_{a}\) is a linear order.

Using the above binary relation $\succeq_a$, for any \(S \in \mathcal{X}\), \(A \in \mathcal{A}\), and  \(x \in S\), we define 
\[
A_{x}(S):=\{a \in A: x\succeq_{a}y, \forall y \in S\}
\]
 and 
\[
A_{x}^{c}(S):=\{a \in A: \exists y \in S \backslash\{x\} \text{ such that } y \succeq_{a}x\}.
\]
Using \textbf{BD}, it can be verified that $A_{x}(S)\cap A_{x}^{c}(S)=\emptyset$. Since $A= A_{x}(S)\cup A_{x}^{c}(S)$, by \textbf{IN}, 
we have 
\begin{displaymath}
    p(x,S,A)=p(x,S,A_{x}(S))p(x,S,A_{x}^{c}(S)) \text{ for all } x\in S.
\end{displaymath} 

We now decompose the terms, $p(x, S, A_{x}(S))$ and $p(x, S, A_{x}^{c}(S))$ separately. By repeatedly applying  \textbf{IN}, \begin{displaymath}
    p(x,S,A_{x}(S))=\prod_{a \in A_{x}(S)}p(x,S,\{a\}) \text{ and } 
    p(x,S,A_{x}^{c}(S))=\prod_{b \in A_{x}^{c}(S)}p(x,S,\{b\}).
\end{displaymath} 

Therefore, 
\begin{displaymath}
    p(x,S,A)=\prod_{a \in A_{x}(S)}p(x,S,\{a\})\prod_{b \in A_{x}^{c}(S)}p(x,S,\{b\}).
\end{displaymath}

 Let $S\backslash\{x\}=\{y_{1},y_{2},\cdots,y_{|S|-1}\}$. Suppose $a\in A_{x}(S)$. By definition of $A_{x}(S)$, $p(x, \{x,y_{k}\} , \{a\})\geq p(y_{k}, \{x,y_{k}\},  \{a\})$ for all $k\in\{1,2,\cdots, |S|-1\}$. Moreover, since $x\notin \{y_{1},\cdots,y_{|S|-1}\}$, by \textbf{BD}, $p(x, \{x,y_{k}\} , \{a\})> p(y_{k}, \{x,y_{k}\},  \{a\})$ for all $k\in\{1,2,\cdots, |S|-1\}$. Therefore, $x\succ_{a}y_{k}$ for all $k\in \{1,2,\cdots,|S|-1\}$. By repeatedly applying \textbf{DOM},
\[p(x,S,\{a\})=p(x,S\backslash\{y_{1}\},\{a\})= p(x, S\backslash\{y_{1},y_{2}\},\{a\})=\cdots =p(x,\{x\},\{a\}).\;\; (\ast)\]

By \textbf{IS}, for any $a\in A$, we can define $\gamma(a)$ as $\gamma(a):=p(w,\{w\},\{a\})$ for any $w\in \mathbf X$. Therefore, by \textbf{IN} and Equation $(\ast)$,
\[p(x, S, A_{x}(S))=\prod_{a\in A_{x}(S)}p(x, \{x\}, \{a\})=\prod_{a\in A_{x}(S)}\gamma(a).\]

Suppose $A_{x}^{c}(S)=\{a_{1},a_{2},\cdots,a_{|A_{x}^{c}(S)|}\}$. Let $y_{1}=\max(\succeq_{a_{1}}, S)$ be the most preferred alternative in $S$ according to $\succeq_{a_{1}}$, i.e., $y\succeq_{a_{1}} y'$ for all $y'\in S$. Similarly, let $y_{k}=\max(\succeq_{a_{k}}, S)$ for all $k\in\{1,2,\cdots,|A_{x}^{c}(S)|\}$. Applying \textbf{FS}, 

\begin{displaymath}
    p(y_{1}, S, \{a_{1}\})+\sum_{z \in S \backslash\{y_{1}\}}p(z,S,\{a_{1}\})=1
\end{displaymath}

 By \textbf{UIA}, for any $x,z\in S\setminus \{y_1\}$, $[y_{1}\succ_{a_{1}}x \text{ and }  y_{1}\succ_{a_{1}}z]\implies [p(x, S, \{a_{1}\})=p(z, S,  \{a_{1}\})]$. Therefore, 
 
 \begin{displaymath}
    p(y_{1}, S, \{a_{1}\})+(|S|-1)p(x,S,\{a_{1}\})=1.
\end{displaymath}

This implies that for $x\in S\setminus \{y_{1}\}$ such that $y_{1}\succ_{a_{1}}x$, we have 

\begin{displaymath}
    p(x,S,\{a_{1}\})= \dfrac{1-p(y_{1},S,\{a_{1}\})}{|S|-1}.
\end{displaymath}

Similarly, we can show that \begin{displaymath}
    p(x,S,\{a_{k}\})= \dfrac{1-p(y_{k},S,\{a_{k}\})}{|S|-1}  \;\;\;  \text{ for all } k\in \{1,2,\cdots,|A_{x}^{c}(S)|\}.
\end{displaymath} 
 Therefore, 

\begin{displaymath}
    p(x,S,\{a_{k}\})=\dfrac{1-\gamma(a_{k})}{(|S|-1)}, \text{ for all } a_{k} \in A_{x}^{c}(S).\;\;\;\;(\dagger) 
\end{displaymath}

By \textbf{IN} and the above equation $(\dagger)$,
\begin{displaymath}
    p(x,S,A_{x}^{c}(S))=\dfrac{\prod_{b \in A_{x}^{c}(S)}(1-\gamma(b))}{(|S|-1)^{|A_{x}^{c}(S)|}}.
\end{displaymath}
 
Therefore,
\begin{displaymath}
    p(x,S,A)=p(x, S, A_{x}(S)) p(x, S, A_{x}^{c}(S))=\prod_{a\in A_{x}(S)}\gamma(a)\dfrac{\prod_{b \in A_{x}^{c}(S)}(1-\gamma(b))}{(|S|-1)^{|A_{x}^{c}(S)|}}.
\end{displaymath}

If \(A=A_{x}(S)\), by \textbf{IN},
\begin{displaymath}
    p(x,S,A)=p(x,S,A_{x}(S))=\prod_{a\in A_{x}(S)}\gamma(a).
\end{displaymath}
Similarly, if \(A=A_{x}^{c}(S)\), by \textbf{IN},
\begin{displaymath}
    p(x,S,A)=p(x,S,A_{x}^{c}(S))=\dfrac{\prod_{b \in A_{x}^{c}(S)}(1-\gamma(b))}{(|S|-1)^{|A_{x}^{c}(S)|}}.
\end{displaymath}

We now show that $\gamma(a)>\dfrac{1}{2}$ for all $a\in \mathbf A$. By \textbf{BD}, for any $x,y\in \mathbf X$ with $x\neq y$, we have $p(x,\{x,y\},\{a\})\neq p(y,\{x,y\},\{a\})$. W.l.o.g. suppose $p(x, \{x,y\},\{a\})>p(y,\{x,y\},\{a\})$. By \textbf{FS}, $p(x, \{x,y\}, \{a\})=1-p(y,\{x,y\},\{a\})>\dfrac{1}{2}.$ Therefore, $\gamma(a)>\dfrac{1}{2}$ for all $a\in \mathbf A$. Moreover, we have shown that $p=p^{s}$ is a MAR with $\succeq_{a}$ and $\gamma(a)$ as parameters for any $a\in \mathbf A$. This completes the proof for the ``Only If" part of the theorem.

$(\impliedby)$ Now, we prove the ``If" part for each Axiom separately. We show that any SCR, $p^{s}$ which is a MAR satisfies all the Axioms mentioned in the statement of Theorem \ref{thm1}.

\begin{claim}
    MAR satisfies \textbf{IN}.
\end{claim}
\medskip
\noindent
{\bf Proof.} Let $p^{s}$ be a MAR. Suppose \(A,B \in \mathcal{A}\) and \(A\cap B= \emptyset\). Assuming $A_x(S)\neq \emptyset$ and $B_x(S)\neq \emptyset$,
for any $x\in S$, we have 

\begin{displaymath}
        p^{s}(x,S, A \cup B)= \prod_{a \in (A \cup B)_{x}(S)} \gamma(a) \dfrac{\prod_{b \in (A \cup B)_{x}^{c}(S) }(1-\gamma(b))}{(|S|-1)^{|(A\cup B)_{x}^{c}(S)|}}
    \end{displaymath}
\begin{displaymath}
    = \prod_{a \in A_{x}(S)}\gamma(a)\prod_{a \in B_{x}(S)}\gamma(a)\left(\dfrac{\prod_{b \in A_{x}^{c}(S)}(1-\gamma(b))}{(|S|-1)^{|A_{x}^{c}(S)|}}\right)\left(\dfrac{\prod_{b \in B_{x}^{c}(S)}(1-\gamma(b))}{(|S|-1)^{|B_{x}^{c}(S)|}}\right)
\end{displaymath}
\begin{displaymath}
    =\left[\prod_{a \in A_{x}(S)}\gamma(a)\dfrac{\prod_{b \in A_{x}^{c}(S)}(1-\gamma(b))}{(|S|-1)^{|A_{x}^{c}(S)|}}\right]\left[\prod_{a \in B_{x}(S)}\gamma(a)\dfrac{\prod_{b \in B_{x}^{c}(S)}1-\gamma(b)}{(|S|-1)^{B_{x}^{c}(S)|}}\right]
\end{displaymath}
\begin{displaymath}
    =p^{s}(x,S,A) p^{s}(x,S,B)
\end{displaymath}
The other cases ($A_x(S)=A$ or $A_x^c(S)=A$) are much simpler and easy to verify. 
\begin{claim}
    MAR satisfies \textbf{FS}.
\end{claim}

\medskip
\noindent
{\bf Proof.} Let $p^{s}$ be a MAR. Since $S$ has finitely many elements, we can list \(S\) as \(S=\{x_{1},\cdots,x_{|S|}\}\). 
As \(\succeq_{a}\) is a linear order over \(S\), w.l.o.g we can order \(S\) as \(x_{1} \succeq_{a}x_{2}\succeq_{a}
\cdots \succeq_{a}x_{|S|}\). Then, \(a \in A_{x_{1}}(S)\) and \(a \in A_{x_{i}}^{c}(S), \forall i \in \{2,\cdots,|S|\}\). By the definition of MAR, \begin{displaymath}
    p^{s}(x_{1},S,\{a\})=\gamma(a)
\end{displaymath}
and
\begin{displaymath}
    p^{s}(x_{i},S,\{a\})=\dfrac{1-\gamma(a)}{|S|-1}, \text{ for all } i \in \{2,\cdots, |S|\}.
\end{displaymath}

Therefore,
\begin{align*}
\sum_{x \in S}p^{s}(x,S,\{a\})&=p^{s}(x_{1},S,\{a\})+\sum_{i=2}^{|S|}p^{s}(x_{i},S,\{a\})\\
    &=\gamma(a)+\sum_{i=2}^{|S|}\dfrac{1-\gamma(a)}{|S|-1}\\
    &=\gamma(a)+(|S|-1)\dfrac{1-\gamma(a)}{(|S|-1)}\\
    &=\gamma(a)+ 1-\gamma(a)\\
    &=1
\end{align*}

\begin{claim} MAR satisfies \textbf{IS}.
\end{claim}

\medskip
\noindent
{\bf Proof.} Let $p^{s}$ be a MAR. Let $x,y\in \mathbf X$ and $a\in \mathbf A$. Letting $A=\{a\}$, we have $A_{x}(\{x \})=A$ since $x\succeq_{a}x$ and $A_{y}(\{y\})=A$ since $y\succeq_{a}y$. By the definition of MAR, 
\[
p^{s}(x,\{x\},\{a\})=\gamma(a)=p^{s}(y,\{y\},\{a\}).
\]

\begin{claim}
    MAR satisfies \textbf{UIA}.
\end{claim}

\medskip
\noindent
{\bf Proof.} Consider any $S\in \mathcal{X}$ and $x,y,z\in S$ such that,  $$p^{s}(x,\{x,y\},\{a\}) > p^{s}(y,\{x,y\}, \{a\}) \text{ and } p^{s}(x,\{x,z\},\{a\}) > p^{s}(z,\{x,z\}, \{a\}). $$ 
By the definition of MAR, $p^{s}(x,V,\{a\})\in\{\gamma(a),1-\gamma(a)\}$, where $V \in \mathcal{X}$ and $|V|=2$. Since, $\gamma(a)>\dfrac{1}{2},\forall a \in \mathbf{A}$, $p^{s}(x,\{x,y\},\{a\})=\gamma(a)=p^{s}(x,\{x,z\},\{a\})$. Similarly, $p^{s}(y,\{x,y\},\{a\})=1-\gamma(a)=p^{s}(z,\{x,z\},\{a\})$.

This implies that $x\succ_{a}y$ and $x\succ_{a}z$. Therefore, by definition of MAR, $$p^{s}(y, S,\{a\})=\dfrac{\gamma(a)}{|S|-1}=p^{s}(z, S,\{a\}).$$

\begin{claim}
    MAR satisfies \textbf{DOM}.
\end{claim}

\medskip
\noindent
{\bf Proof.} Let $p^{s}$ be a MAR. Consider any $a\in \mathbf A$, $S\in\mathcal{X}$ and $x\in S$ such that 
$p^{s}(x,\{x,y\},\{a\})>p^{s}(y,\{x,y\},\{a\})$ for all $y\in S\setminus \{x\}$. This implies by MAR that $x\succ_{a}y$
for all $y\in S\setminus \{x\}$. Therefore, 
      \[p^{s}(x,S,\{a\})=\gamma(a)=p^{s}(x,S\backslash\{y\},\{a\}).\]

\begin{claim}
    MAR satisfies \textbf{BD}.
\end{claim}

\medskip
\noindent
{\bf Proof.} Let $p^{s}$ be a MAR. Consider any $a\in \mathbf A$ and $x,y\in \mathbf X$ with $x\neq y$. 
Since \(\succeq_{a}\) is a linear order over \(X\), either \(x \succ_{a} y\) or \(y \succ_{a} x\). W.l.o.g. suppose \(x \succ_{a} y\). By the definition of MAR,
    \begin{displaymath}
        p^{s}(x,\{x,y\},\{a\})=\gamma(a)>\dfrac{1-\gamma(a)}{2-1}=1-\gamma(a)=p^{s}(y,\{x,y\},\{a\}).
    \end{displaymath}
    
    Therefore, \(p^{s}(x,\{x,y\},\{a\}) \not = p^{s}(y,\{x,y\}, \{a\})\). 

\begin{claim}
    MAR satisfies \textbf{ST}.
\end{claim}

\medskip
\noindent
{\bf Proof.} Let $p^{s}$ be a MAR.  Suppose for some $x,y,z\in X$, 
$$p^{s}(x,\{x,y\},\{a\}) > p^{s}(y,\{x,y\},\{a\}) \text { and } p^{s}(y,\{y,z\},\{a\}) > p^{s}(y,\{y,z\},\{a\}).$$ 

Since $p^{s}$ is a MAR with respect to $\succeq_{a}$ for any $a\in A$, the above two strict inequalities imply (as shown earlier) that $x\succ_{a}y$ and $y\succ_{a}z$. Since $\succ_{a}$ is transitive, $x\succ_{a}z$ which implies that $p(x,\{x,z\},\{a\})>p(z,\{x,z\},\{a\})$. \end{proof}

\subsection{Uniqueness Of Parameters}\label{sec4.1}

We now establish the uniqueness of parameters in the model which concern the salience of attributes and the attribute-based set of preferences over the alternatives, namely, \((\gamma(a))_{a\in \mathbf A}\) and  \((\succeq_{a})_{a\in \mathbf A}\).

\begin{proposition}
    The parameters \(((\gamma(a))_{a\in \mathbf A},(\succeq_{a})_{a\in \mathbf A})\) for a MAR $p^{s}$ are unique\footnote{That is, if \(((\gamma(a))_{a\in \mathbf A},(\succeq_{a})_{a\in \mathbf A})\) and \(((\gamma^{'}(a))_{a\in 
\mathbf A},(\succeq^{'}_{a})_{a\in \mathbf A})\) both represent 
$p$ under MAR, then $\gamma(a)=\gamma^{'}(a)$ and $\succeq_a=\succeq^{'}_a,$ $\forall a \in \mathbf{A}$. } for any given choice data.
\end{proposition}
\begin{proof} Let $p^{s}$ be a MAR that represents a choice data $p$. 
    Suppose the family of parameters \(((\gamma(a))_{a\in \mathbf A},(\succeq_{a})_{a\in \mathbf A})\) and \(((\gamma^{'}(a))_{a\in 
\mathbf A},(\succeq^{'}_{a})_{a\in \mathbf A})\) both represent 
$p$ under MAR $p^s$. We show that \(((\gamma(a))_{a\in \mathbf A},(\succeq_{a})_{a\in \mathbf A})\)\(=((\gamma^{'}(a))_{a\in \mathbf A},(\succeq^{'}_{a})_{a\in 
\mathbf A})\). For any $x\in 
\mathbf X$, by definition of MAR, we have  
\[
p(x,\{x\},\{a\})=\gamma(a)=\gamma^{'}(a)=p(x,\{x\},\{a\}), \forall \, a \in \mathbf A.
\] 
Thus, \((\gamma(a))_{a\in \mathbf A}=(\gamma^{'}(a))_{a\in \mathbf A}\). 
    
 Now, suppose that \((\succeq_{a})_{a\in \mathbf A}\not = (\succeq^{'}_{a})_{a\in \mathbf A} \). Then, \(\exists \, x,y \in X\) with $x\neq y$ and \(\exists \, a \in A\) such that \(x \succeq_{a}y\) and \(\ y\succeq^{'}_{a}x\). By MAR and by the fact that \(\gamma(a)\) is unique, we have 
\[
p(x,\{x,y\},\{a\})=\gamma(a)=\dfrac{1-\gamma(a)}{2-1}=1-\gamma(a).
\]
 This implies that \(\gamma(a)=\dfrac{1}{2}\). However, this is a contradiction to the definition of MAR since $\gamma(a)>\dfrac{1}{2}$ for all $a \in \mathbf A$. Hence, 
\((\succeq_{a})_{a\in \mathbf A}= (\succeq^{'}_{a})_{a\in {\mathbf A}}\). 
\end{proof}

\subsection{Regularity in Alternatives and in Attributes}\label{sec4.2}
We first show that MAR is regular in attributes.

\begin{definition}[Regular in Attributes]
    An SCR $p$ is said to be Regular in Attributes if \(p(x,S,A) \geq p(x,S,B)\), whenever \(A, B \in \mathcal{A}\) with $A\subseteq B$, and $x \in S$.
\end{definition}

An interpretation of regularity in attributes is as follows. As more attributes become relevant, the choices become more difficult and the probability that an alternative is chosen weakly decreases (hence the probability that the empty set or the default alternative is chosen weakly increases) when the set of alternatives is fixed. 

\begin{proposition}\label{reginatt}
    MAR satisfies Regularity in Attributes.
\end{proposition}
\begin{proof} Suppose $p^{s}$ is a MAR and let $S\in\mathcal{X}$ and \(A, B \in \mathcal{A}\) with $A\subseteq B$. 
  By finiteness of \(B\), we can write \(B= A \cup \{b_{1},\cdots,b_{k}\}\) for some \(k \in \mathbf{N}\). Denote \(B_{1}=A \cup \{b_{1}\}\). Therefore, 
\begin{displaymath}
        p^{s}(x,S,B_{1})= \begin{cases}
            \ \gamma(b_{1})p^{s}(x,S,A) &\quad \text{ if } b_{1} \in  B_{1_{x}}(S); \\
            \ \left(\dfrac{1-\gamma(b_{1})}{|S|-1}\right)p^{s}(x,S,A) &\quad \text{ if } b_{1} \in B_{1_{x}}^{c}(S).
        \end{cases}
    \end{displaymath}
Since \(\gamma(b_{1}) \in (\dfrac{1}{2},1]\), we have \(p^{s}(x,S,B_{1}) \leq p^{s}(x,S,A)\). Similarly, one can show that 
\[
p(x,S,B_{k})\leq \cdots \leq p(x,S,B_{2}) \leq p(x,S,B_{1}) \leq p^{s}(x,S,B_{0})= p^{s}(x,S,A),
\] 
where \(B_{0}=A\) and \(B_{i}=B_{i-1} \cup \{b_{i}\}, \forall i \in \{1,\cdots,k\}\). Thus, 
\[
B_{k}=B_{k-1} \cup \{b_{k}\}=A \cup \{b_{1},\cdots,b_{k-1}\} \cup \{b_{k}\}=B.
\]
Consequently, \(p^{s}(x,S,A) \geq p^{s}(x,S,B)\). This establishes the result.
\end{proof}

 \begin{definition}[Regularity in Alternatives]
     An SCR $p$ satisfies Regularity in Alternatives if $p(x,S, A) \geq p(x,T, A)$, for all $A\in\mathcal{A}$ and for all $x \in S \subseteq T$, where $S, T \in \mathcal{X}$.
 \end{definition}
This is the standard Regularity postulate which states that the choice probability of an alternative is weakly greater for a subset of a set when the set of attributes is fixed.

\begin{proposition}
    MAR satisfies Regularity in alternatives.
\end{proposition}
\begin{proof} Let \(S, T \in \mathcal{X}\) with $S\subseteq T$. Assuming $A_x(S)\neq \emptyset$ and $A_{x}^{c}(S)\neq \emptyset$,
for any $x\in S$, by MAR, we have  
\[
p^{s}(x,S,A)=\prod_{a \in A_{x}(S)}\gamma(a)\dfrac{\prod_{a \in A^{c}_{x}(S)}(1-\gamma(a))}{(|S|-1)^{|A_{x}^{c}(S)|}}
\]
 for any $A\in \mathcal{A}$. Since, \(T\) is finite, we can write \(T= S \cup \{x_{1},\cdots,x_{k}\}\) for some \(k \in \mathbf{N}\). Denote \(T_{1}=S \cup \{x_{1}\}\). By the construction of the sets \(A_{x}(\cdot)\) and \(A_{x}^{c}(\cdot)\), for any attribute \(a \in A\) it is possible that \(a \in A_{x}(S)\) and \(a \in A_{x}^{c}(S \cup \{x_{1}\})\) (for instance if \(\exists \; a \in A\) such that \(x \succeq_{a} y, \forall y \in S\) but \(x_{1} \succeq_{a} x\)). However, it is not possible that \(a \in A_{x}^{c}(S)\) and \(a \in A_{x}(S \cup \{x_{1}\})\) (even if \(x \succeq_{a} x_{1}, \forall a \in A\)). Combining the above with the fact that 
\[
\gamma(a) > \dfrac{1-\gamma(a)}{|S|-1} \geq \dfrac{1-\gamma(a)}{|S|}=\dfrac{1-\gamma(a)}{|S\cup \{x_{1}\}|-1}, \forall a \in A\]
 (the inequality is strict for \(\gamma(a) \in [ \dfrac{1}{2},1)\)) we have \(p^{s}(x,S,A) \geq p^{s}(x,S \cup \{x_{1}\}, A)\). Note that $p^{s}(x,S,A) = p^{s}(x,S \cup \{x_{1}\}, A)$ holds only either when \(x \succeq_{a} x_{1}, \forall a \in A\) and \(A_{x}^{c}(S)=\emptyset\)\footnote{Since, \[
\dfrac{1-\gamma(a)}{|S|-1} \geq \dfrac{1-\gamma(a)}{|S \cup \{x_{1}\}|-1}
 \text{ (equality only when } \gamma(a)=1).\] } or when $\exists b\in A_{x}^c(S) \text{ such that } \gamma(b)=1 \text{ (then } p^{s}(x,S,A) = p^{s}(x,S \cup \{x_{1}\}, A)=0)$.

Thus, $p^{s}(x,S,A) \geq p^{s}(x,S \cup \{x_{1}\}, A)$.
 Similarly, one can show that 
\[
p^{s}(x,T,A)=p^{s}(x,T_{k},A)\leq \cdots \leq p^{s}(x,T_{1},A) \leq p^{s}(x,T_{0},A)= p^{s}(x,S,A),
\] 
where \(T_{0}=S\) and \(T_{i}=T_{i-1} \cup \{x_{i}\}, \forall i \in \{1,\cdots,k\}\). Thus, 
\[
T_{k}=T_{k-1} \cup \{x_{k}\}=S \cup \{x_{1},\cdots,x_{k-1}\} \cup \{x_{k}\}=T.
\] 
Therefore, \(p^{s}(x,S,A) \geq p^{s}(x,T,A)\) whenever \(S \subseteq T \in \mathcal{X}\).
\end{proof}

 MAR does not satisfy a stronger notion of regularity, which can be defined as follows.
 \begin{definition}[Regularity in Alternatives*] An SCR $p$ satisfies Regularity in Alternatives* if $p(x,S,A)\geq p(x,T,B)$ for all $x\in S\subseteq T\in\mathcal{X}$ and $A,B\in \mathcal{A}$.\end{definition}
 
 Therefore, this version of regularity allows the attribute sets to be different. We show an example to show that MAR violates this stronger version of regularity in alternatives. Let \(x \succeq_{a} z  \succeq_{a}y\),  \(y \succeq_{b} x  \succeq_{b}z\), \(x \succeq_{c} y  \succeq_{c}z\) and \(\gamma(a)=0.6,\gamma(b)=0.7,\gamma(c)=0.9\). Then by MAR, \(p^{s}(x,\{x,y\},\{a,b\})=0.18 <p^{s}(x, \{x,y,z\}, \{a,c\})=0.54\).

In the above example, we can make another observation that, \textit{when the attribute set is varied}, MAR can explain Asymmetric Dominance (or Attraction Effect) (see \cite{huber1982adding}) which is described as follows: when the asymmetrically dominated element \(z\) is added to \(\{x,y\}\), the choice probability of dominant element \(x\) is the highest in \(\{x,y,z\}\) and in particular there is a reversal in choice probability ordering between \(x\) and \(y\) as the following calculations show: 
\begin{displaymath}
    p^{s}(x,\{x,y\},\{a,b\})=0.18 < p^{s}(y, \{x,y\}, \{a,b\})=0.28 \end{displaymath} and \begin{displaymath} p^{s}(x, \{x,y,z\}, \{a,c\})=0.54 > p^{s}(y, \{x, y, z\}, \{a,c\})=0.01.
\end{displaymath}
\begin{remark}
    Note that in our example, it is crucial that the attributes sets are varied across contexts when the asymmetrically dominated alternative is added to the initial menu. We view this as a limitation and this motivates our next rule.
\end{remark}


\section{The Additive Attention Rule}\label{secAAR}

MAR aggregates attention multiplicatively across attributes and, within a single attribute, is tops-only: the $\succeq_a$-maximal alternative receives $\gamma(a)$ and every other alternative receives $\frac{1-\gamma(a)}{|S|-1}$, uniformly. As a consequence, MAR is regular in alternatives for a fixed attribute set, and it reproduces the Attraction Effect only when the attribute set is allowed to vary across menus, as in the example of Section \ref{sec4.2}.

In this section we isolate the additive counterpart of this mechanism. We introduce the Additive Attention Rule (AAR), which averages attention across attributes and, within an attribute, weights an alternative by a normalized measure of its rank rather than by just whether it is top-ranked. This single change of aggregation is exactly what lets the rule generate the Attraction Effect with the attribute set held fixed. We characterize the rule with six axioms,  three of which (IS, BD, ST) are inherited from Theorem \ref{thm1}, and three of which (AVG, MA, BC) are new, and show that, while it accommodates the Attraction Effect, it provably cannot accommodate the Compromise Effect.

\setcounter{definition}{5}
\subsection{The rule}\label{secaar5.1}

For an attribute $a$ with linear order $\succeq_a$ and a menu $S$, let
\[
V(x,a,S)=\bigl|\{y\in S: x\succeq_a y\}\bigr|
\]
denote the (weak) rank\footnote{We could also use the term `score' instead of `rank'.} of $x$ in $S$ under $a$, so that the $\succeq_a$-top alternative has rank $|S|$, the $\succeq_a$-bottom alternative has rank $1$, and $V(x,a,\{x\})=1$. Write
\[
N(S):=\sum_{x\in S}V(x,a,S)=\frac{|S|(|S|+1)}{2}
\]
for the rank sum of a menu of size $|S|$ --- this value is the same for every attribute $a$, since $V(\cdot,a,S)$ is a bijection from $S$ onto $\{1,\dots,|S|\}$ for each $a$.

\begin{definition}\label{aardef}[Additive Attention Rule (AAR)]
An SCR $p$ is an Additive Attention Rule if there exist, for each $a\in \mathbf A$, a linear order $\succeq_a$ on $\mathbf{X}$ and an attention parameter $\gamma(a)\in(0,1]$ such that for all $S\in \mathcal{X}$, $A\in\mathcal{A}$, and $x\in S$,
\[
p(x,S,A)=\frac{1}{N(S)\,|A|}\sum_{a\in A}\gamma(a)\,V(x,a,S)=\frac{2}{|S|(|S|+1)\,|A|}\sum_{a\in A}\gamma(a)\,V(x,a,S).
\]
\end{definition}

In words, $x$'s choice probability is the attention-weighted average, across attributes, of $x$'s own rank within the menu, normalized by the rank sum $N(S)$. An alternative earns credit for \emph{every} attribute under which it is not at the bottom, in proportion to how far up that attribute's ranking it sits and how much attention the attribute draws --- not only for the attributes under which it happens to be ranked first. An alternative that is a strong second choice everywhere can therefore out-poll one that is first under a single, weakly attended attribute, a possibility MAR's tops-only aggregation rules out by construction. Unlike MAR, the probability tracks the entire rank profile, not only the identity of the top alternative; this is the source of the behaviour below.

\begin{example}\label{egAdditive}
 Take the menu and orderings of Example \ref{egA} --- $S=\{x,y,z\}$, $x\succ_{a_1}y\succ_{a_1}z$, $y\succ_{a_2}x\succ_{a_2}z$, $\gamma(a_1)=0.6$, $\gamma(a_2)=0.9$ --- and evaluate them under the Additive Attention Rule instead of MAR. Here $N(S)=6$, and the ranks are $V(x,a_1,S)=3$, $V(y,a_1,S)=2$, $V(z,a_1,S)=1$ and $V(y,a_2,S)=3$, $V(x,a_2,S)=2$, $V(z,a_2,S)=1$, so
\[
p(x,S,A)=\tfrac1{12}\bigl[0.6(3)+0.9(2)\bigr]=0.3,\qquad p(y,S,A)=\tfrac1{12}\bigl[0.6(2)+0.9(3)\bigr]=0.325,
\]
\[
p(z,S,A)=\tfrac1{12}\bigl[0.6(1)+0.9(1)\bigr]=0.125,\qquad p(x^*,S,A)=1-0.75=0.25.
\]
\end{example}

The contrast with Example \ref{egA} is instructive. There, under MAR, $y$'s probability (0.18) was six times $x$'s (0.03) and eighteen times $z$'s (0.01), and the default absorbed 0.78 of the mass, because $x$ and $z$ earned nothing for their standing under $a_1$ once $y$ won under the more-attended $a_2$: only the top-ranked alternative under each attribute is ever rewarded. Under the Additive Attention Rule, by contrast, $x$'s strong showing under $a_1$ --- where it is top-ranked and attention is a non-trivial $0.6$ --- still earns it substantial credit even though $y$ wins under $a_2$: $y$ edges out $x$ only narrowly ($0.325$ against $0.3$), $z$ --- last under both attributes --- still receives a modest $0.125$, and the default's share falls to a quarter rather than three-quarters. Being well-ranked, rather than merely being on top, is what keeps the middle and bottom of the menu in contention.

\subsection{The Attraction Effect}\label{secATTAAR}

The Attraction Effect (\cite{huber1982adding}) is the reversal of a binary choice ranking once a decoy is added to the menu: an alternative that is dominated by one of the two original options but not by the other, so that it is asymmetrically, rather than uniformly, inferior. Section \ref{sec4.2} showed that MAR can reproduce this only by letting the relevant attribute set itself change as the decoy is added. We now show that AAR reproduces it with the attribute set held fixed throughout.\footnote{It is obvious that AAR can also reproduce this even when the relevant attribute set can change as the decoy is added.}

The argument has two parts. We first show that under a fixed pair of attributes, whichever of $x$ and $y$ is favoured by the more heavily attended attribute wins the binary comparison $\{x,y\}$, this is unsurprising and holds for MAR too. The content is in the second part: introducing a decoy $z$ that $x$ dominates under both attributes, but that itself beats $y$ under one of them, changes the \emph{ranks} $x$ and $y$ receive under both attributes, not just their pairwise standing against $z$. Because AAR credits full rank rather than only top position, this rank shift is enough to overturn the binary ranking of $x$ against $y$, with the same two attributes in play throughout.

Fix the attribute set $A=\{a,b\}$ and suppose
\[
x\succ_a z\succ_a y \qquad\text{and}\qquad y\succ_b x\succ_b z.
\]
In the binary menu $\{x,y\}$, $N(\{x,y\})=3$; under $a$ the ranks are $V(x,a,\{x,y\})=2$, $V(y,a,\{x,y\})=1$, and under $b$ they are $V(x,b,\{x,y\})=1$, $V(y,b,\{x,y\})=2$, so
\[
p(x,\{x,y\},A)=\tfrac16\bigl(2\gamma(a)+\gamma(b)\bigr),\qquad p(y,\{x,y\},A)=\tfrac16\bigl(\gamma(a)+2\gamma(b)\bigr),
\]
and hence
\[
p(x,\{x,y\},A)<p(y,\{x,y\},A)\iff\gamma(a)<\gamma(b).
\]
With attention ordered $\gamma(a)<\gamma(b)$, alternative $y$ is chosen more often than $x$ in the binary menu.

Now introduce the decoy $z$, dominated by $x$ under both attributes ($x\succ_az$ and $x\succ_bz$) but not by $y$ (indeed $z\succ_ay$). In $\{x,y,z\}$, $N(\{x,y,z\})=6$, and the ranks become $V(x,a,\{x,y,z\})=3$, $V(y,a,\{x,y,z\})=1$, $V(x,b,\{x,y,z\})=2$, $V(y,b,\{x,y,z\})=3$, so
\[
p(x,\{x,y,z\},A)=\tfrac1{12}\bigl(3\gamma(a)+2\gamma(b)\bigr),\qquad p(y,\{x,y,z\},A)=\tfrac1{12}\bigl(\gamma(a)+3\gamma(b)\bigr),
\]
and therefore
\[
p(x,\{x,y,z\},A)>p(y,\{x,y,z\},A)\iff\gamma(a)>\tfrac{\gamma(b)}{2}.
\]
Combining the two displays, the choice ranking of $x$ and $y$ reverses upon adding the dominated decoy --- the AAR exhibits the Attraction Effect (\cite{huber1982adding}) --- precisely when
\[
\frac{\gamma(b)}{2}<\gamma(a)<\gamma(b),
\]
with the attribute set held fixed at $\{a,b\}$ throughout. This is the qualitative departure from MAR, whose regularity in alternatives forces any such reversal to be driven instead by a change in the relevant attributes.

\setcounter{axiom}{7}
\subsection{Axiomatic characterization}\label{sec.AARaxiomaticchar}

The AAR retains IS, BD, and ST from Section 3 and replaces the remaining MAR axioms with three. The first, AVG, is the additive analogue of IN: attribute aggregation is by averaging rather than by product. The second and third, MA and BC, encode the within-attribute rank structure and play the role that FS, UIA, and DOM play for MAR.

\begin{axiom}[Attribute Average (AVG)]
An SCR $p$ satisfies AVG if for all $x\in S\in\mathcal{X}$ and all $A\in\mathcal{A}$,
\[
p(x,S,A)=\frac{\sum_{a\in A}p(x,S,\{a\})}{|A|}.
\]
\end{axiom}

\begin{axiom}[Menu Additivity (MA)]
An SCR $p$ satisfies MA if for all $a\in \mathbf A$ and all $S,T\in\mathcal{X}$ with $S\cap T=\{x\}$ for some $x\in \mathbf{X}$,
\[
N(S\cup T)\,p(x,S\cup T,\{a\})=N(S)\,p(x,S,\{a\})+N(T)\,p(x,T,\{a\})-p(x,\{x\},\{a\}),
\]
where $N(S)=|S|(|S|+1)/2$.
\end{axiom}

To understand MA, take a singleton attribute set $\{a\}$ and a menu $S$ containing $x$, together with a second menu $T$ that also contains $x$ but otherwise shares nothing with $S$. MA states that $x$'s weighted choice probability in the merged menu $S\cup T$ is the sum of its weighted choice probabilities in $S$ and $T$ separately, less one copy of $p(x,\{x\},\{a\})$ to avoid double-counting; the weight attached to a menu of size $n$ is the rank sum $N(n)=n(n+1)/2$.

\begin{axiom}[Binary Calibration (BC)]
An SCR $p$ satisfies BC if for all $a\in \mathbf A$ and all $x,y\in \mathbf{X}$ with $x\ne y$, whenever $p(x,\{x,y\},\{a\})\ge p(y,\{x,y\},\{a\})$:
\[
p(x,\{x,y\},\{a\})=\tfrac23\,p(x,\{x\},\{a\}),\qquad p(y,\{x,y\},\{a\})=\tfrac13\,p(y,\{y\},\{a\}).
\]
\end{axiom}

BC pins the exact exchange rate between binary and singleton data: whichever alternative (weakly) wins the pairwise comparison under $a$ receives two-thirds of its own singleton value, and the other receives exactly one-third of its own singleton value.

\setcounter{theorem}{1}
\begin{theorem}\label{thmAAR}
An SCR is an Additive Attention Rule if and only if it satisfies AVG, MA, BC, BD, IS, and ST.
\end{theorem}

The necessity direction is immediate substitution into the AAR formula. Sufficiency is more involved: as in Theorem \ref{thm1}, BD and ST first recover the attribute-specific order $\succeq_a$ from binary data, and IS pins down $\gamma(a)$ from singletons. What differs is how we extend from binary menus to general ones. Rather than DOM's alternative-removal argument, MA lets us telescope outward one alternative at a time, adding each new $y_k$ to a growing menu and using BC to check, at each step, whether $y_k$ sits above or below $x$ under $\succeq_a$, and hence how much it contributes to $x$'s weighted choice probability. Summing these contributions across all of $S\setminus\{x\}$ recovers exactly the rank $V(x,a,S)$, which is where the additive, full-rank structure of AAR comes from.

\begin{proof}
\emph{Necessity.} Let $p$ be an AAR with linear orders $(\succeq_a)_{a\in A}$ and weights $\gamma(a)\in(0,1]$; recall $p(x,S,\{a\})=\gamma(a)V(x,a,S)/N(S)$. IS holds since $p(x,\{x\},\{a\})=\gamma(a)V(x,a,\{x\})/N(\{x\})=\gamma(a)\cdot1/1=\gamma(a)$, independent of $x$. BD holds since, for $x\ne y$, $\{V(x,a,\{x,y\}),V(y,a,\{x,y\})\}=\{1,2\}$ and $N(\{x,y\})=3$, so $p(x,\{x,y\},\{a\})-p(y,\{x,y\},\{a\})=\pm\gamma(a)/3\ne0$. ST holds since $p(x,\{x,y\},\{a\})>p(y,\{x,y\},\{a\})\iff x\succ_a y$, and $\succeq_a$ is transitive. AVG holds since \[\frac1{|A|}\sum_{a\in A}p(x,S,\{a\})\]\[=\frac1{|A|}\sum_{a\in A}\gamma(a)V(x,a,S)/N(S)=p(x,S,A)\].

For MA, note that whenever $S\cap T=\{x\}$, the sets $S\setminus\{x\}$ and $T\setminus\{x\}$ are disjoint, so
\begin{align*}
V(x,a,S\cup T)&=\bigl|\{y\in S\cup T:x\succeq_ay\}\bigr|\\
&=\bigl|\{y\in S:x\succeq_ay\}\bigr|+\bigl|\{y\in T:x\succeq_ay\}\bigr|-1\\
&=V(x,a,S)+V(x,a,T)-1,
\end{align*}
the sets on the right overlapping only in $x$ itself. Hence
\begin{align*}
N(S\cup T)p(x,S\cup T,\{a\})&=\gamma(a)V(x,a,S\cup T)=\gamma(a)\bigl[V(x,a,S)+V(x,a,T)-1\bigr]\\
&=N(S)p(x,S,\{a\})+N(T)p(x,T,\{a\})-\gamma(a),
\end{align*}
and since $\gamma(a)=p(x,\{x\},\{a\})$, as desired.

For BC, note that $p(x,\{x,y\},\{a\})\ge p(y,\{x,y\},\{a\})\iff x\succeq_ay\iff V(x,a,\{x,y\})=2$, giving $p(x,\{x,y\},\{a\})=\tfrac23\gamma(a)=\tfrac23p(x,\{x\},\{a\})$; and $V(y,a,\{x,y\})=1$ gives $p(y,\{x,y\},\{a\})=\tfrac13\gamma(a)=\tfrac13p(y,\{y\},\{a\})$, as required.

\emph{Sufficiency.} Let $p$ satisfy the six axioms. By AVG, $p(x,S,A)=\tfrac1{|A|}\sum_{a\in A}p(x,S,\{a\})$, so it suffices to pin down $p(\cdot,\cdot,\{a\})$ for each $a$. For each $a$ define
\[
x\succeq_ay\iff p(x,\{x,y\},\{a\})\ge p(y,\{x,y\},\{a\}).
\]
Completeness is immediate from totality of $\ge$ on $\mathbb R$; antisymmetry follows from BD (no binary ties); transitivity follows from ST. Hence $\succeq_a$ is a linear order, with strict part $\succ_a$. By IS, $\gamma(a):=p(x,\{x\},\{a\})$ is well defined independently of $x\in \mathbf{X}$. For $x\in S\in\mathcal{X}$ define
\[
K(x,S,a):=\bigl|\{y\in S: p(x,\{x,y\},\{a\})\ge p(y,\{x,y\},\{a\})\}\bigr|=\bigl|\{y\in S:x\succeq_ay\}\bigr|=V(x,a,S),
\]
and let $D_x(a):=\{x\}\cup\{y\in \mathbf{X}\setminus\{x\}: p(x,\{x,y\},\{a\})\ge p(y,\{x,y\},\{a\})\}$, so that $K(x,S,a)=|S\cap D_x(a)|$ for every $S\in\mathcal{X}$ with $x\in S$.

Fix $a\in A$ and $S=\{x,y_1,\dots,y_n\}\in\mathcal{X}$. Set $S_0=\{x\}$ and $S_k=S_{k-1}\cup\{x,y_k\}$ for $k=1,\dots,n$. Since the $y_k$'s are distinct, $S_{k-1}\cap\{x,y_k\}=\{x\}$, so MA applies at each step:
\[
N(S_k)p(x,S_k,\{a\})=N(S_{k-1})p(x,S_{k-1},\{a\})+N(\{x,y_k\})p(x,\{x,y_k\},\{a\})-\gamma(a).
\]
By BD, either $p(x,\{x,y_k\},\{a\})> p(y_k,\{x,y_k\},\{a\})$ or the reverse strict inequality holds. In the first case, BC gives $p(x,\{x,y_k\},\{a\})=\tfrac23\gamma(a)$, so $N(\{x,y_k\})p(x,\{x,y_k\},\{a\})-\gamma(a)=3\cdot\tfrac23\gamma(a)-\gamma(a)=\gamma(a)$, and $y_k\in D_x(a)$. In the second, BC applied to the pair $(y_k,x)$ gives $p(x,\{x,y_k\},\{a\})=\tfrac13\gamma(a)$, so $N(\{x,y_k\})p(x,\{x,y_k\},\{a\})-\gamma(a)=3\cdot\tfrac13\gamma(a)-\gamma(a)=0$, and $y_k\notin D_x(a)$. Hence
\[
N(S_k)p(x,S_k,\{a\})-N(S_{k-1})p(x,S_{k-1},\{a\})=\gamma(a)\cdot\mathbf{1}[y_k\in D_x(a)].
\]
Summing over $k=1,\dots,n$, and noting $N(S_0)p(x,S_0,\{a\})=N(\{x\})\gamma(a)=\gamma(a)$ and $|S\cap D_x(a)|=1+\sum_{k=1}^n\mathbf 1[y_k\in D_x(a)]=K(x,S,a)=V(x,a,S)$:
\[
N(S)p(x,S,\{a\})=\gamma(a)+\gamma(a)\sum_{k=1}^n\mathbf1[y_k\in D_x(a)]=\gamma(a)\,V(x,a,S).
\]
Hence $p(x,S,\{a\})=\gamma(a)V(x,a,S)/N(S)$ for every $S\in\mathcal{X}$, and averaging over $a\in A$ via AVG gives
\[
p(x,S,A)=\frac1{|A|}\sum_{a\in A}\frac{\gamma(a)}{N(S)}V(x,a,S)=\frac{2}{|S|(|S|+1)|A|}\sum_{a\in A}\gamma(a)\,V(x,a,S),
\]
which is the Additive Attention Rule.
\end{proof}

\subsubsection*{Uniqueness of Parameters}

As with MAR, the parameters governing the Additive Attention Rule --- the attention probabilities $(\gamma(a))_{a\in \mathbf A}$ and the attribute-based preferences $(\succeq_a)_{a\in \mathbf A}$ --- are uniquely recovered from choice data.

\begin{proposition}
The parameters $\bigl((\gamma(a))_{a\in \mathbf A},(\succeq_a)_{a\in \mathbf  A}\bigr)$ for an Additive Attention Rule $p$ are unique for any given choice data.
\end{proposition}

\begin{proof}
Suppose $\bigl((\gamma(a))_{a\in \mathbf A},(\succeq_a)_{a\in \mathbf A}\bigr)$ and $\bigl((\gamma'(a))_{a\in \mathbf A},(\succeq'_a)_{a\in \mathbf A}\bigr)$ both represent the same choice data $p$ as an Additive Attention Rule. For any $x\in \mathbf X$ and $a\in \mathbf A$, $N(\{x\})=1$ and $V(x,a,\{x\})=1$, so
\[
\gamma(a)=p(x,\{x\},\{a\})=\gamma'(a),
\]
giving $(\gamma(a))_{a\in \mathbf A}=(\gamma'(a))_{a\in \mathbf A}$; write $\gamma(a)$ for this common value.

Now suppose $(\succeq_a)_{a\in \mathbf A}\ne(\succeq'_a)_{a\in \mathbf A}$: there exist $a\in \mathbf A$ and $x\ne y$ with $x\succ_ay$ and $y\succ'_ax$. Since $N(\{x,y\})=3$, evaluating $p(x,\{x,y\},\{a\})$ under the first representation --- where $x$ is top-ranked, so $V(x,a,\{x,y\})=2$ --- gives
\[
p(x,\{x,y\},\{a\})=\frac{\gamma(a)}{3}\cdot2=\frac{2\gamma(a)}{3},
\]
while evaluating the very same quantity under the second representation --- where $x$ is bottom-ranked, so $V(x,a,\{x,y\})=1$ --- gives
\[
p(x,\{x,y\},\{a\})=\frac{\gamma(a)}{3}\cdot1=\frac{\gamma(a)}{3}.
\]
Since both expressions represent the same choice probability, $\tfrac23\gamma(a)=\tfrac13\gamma(a)$, so $\gamma(a)=0$, contradicting $\gamma(a)\in(0,1]$. Hence $(\succeq_a)_{a\in \mathbf A}=(\succeq'_a)_{a\in \mathbf A}$.
\end{proof}

\subsection{Failure of the Compromise Effect}\label{secAARcomp}

Having shown that AAR can reproduce the Attraction Effect, it is natural to ask whether it can also generate the Compromise Effect (\cite{simonson1989choice}): the tendency for a middle-of-the-road alternative to become the modal choice once it is flanked by two more extreme options. We show that it cannot, and that the reason is instructive in its own right.

The obstruction is linearity, not the specific numbers. A compromise alternative $z$ needs to beat two rivals that are each extreme in a different direction, best in one attribute, worst in the other. Under additive rank-weighting, beating $x$ (who is top in $a$) requires attention to $b$ to outweigh attention to $a$ by enough, while beating $y$ (top in $b$) requires the reverse. No single pair of weights can satisfy both simultaneously. The following derivation confirms this.

The Compromise Effect (\cite{simonson1989choice}) requires a ``middle'' alternative to be chosen most often. Suppose
\[
x\succ_az\succ_ay\qquad\text{and}\qquad y\succ_bz\succ_bx,
\]
so that $z$ is the compromise, ranked second under both attributes. The effect would require, with $S=\{x,y,z\}$ and $A=\{a,b\}$,
\[
p(z,S,A)>\max\{p(x,S,A),p(y,S,A)\}.
\]
Here $N(S)=6$, and the ranks are
\[
V(z,a,S)=V(z,b,S)=2,\qquad (V(x,a,S),V(x,b,S))=(3,1),\qquad (V(y,a,S),V(y,b,S))=(1,3),
\]
so
\[
p(z,S,A)=\tfrac16\bigl(\gamma(a)+\gamma(b)\bigr),\qquad p(x,S,A)=\tfrac1{12}\bigl(3\gamma(a)+\gamma(b)\bigr),\qquad p(y,S,A)=\tfrac1{12}\bigl(\gamma(a)+3\gamma(b)\bigr).
\]
Then $p(z,S,A)>p(x,S,A)\iff\gamma(b)>\gamma(a)$, whereas $p(z,S,A)>p(y,S,A)\iff\gamma(a)>\gamma(b)$. The two conditions are mutually exclusive, so no attention weights make $z$ the modal choice: the AAR cannot exhibit the Compromise Effect.

The contrast with the Attraction Effect is worth noting. Attraction is driven by a dominated decoy, which raises the rank of the target under one attribute without raising its rival's; the rule converts this rank gain directly into a score gain. Compromise, by contrast, asks a rank-two alternative to beat two extreme (having been rank-one in one attribute and being rank-three in the other attribute) alternatives simultaneously, which linear, additive aggregation of ranks cannot deliver. Accommodating the Compromise Effect would require a non-linear (e.g.\ concave) transformation of ranks, which lies outside the additive class characterized here.
\begin{remark}
    MAR cannot explain the Compromise Effect in cases where the set of attributes is either fixed or variable. The reason being that MAR is a stochastic version of a tops-only rule (see Appendix \ref{sec7}) and a Compromise alternative is never at the top of any attribute. 
\end{remark}

\subsection{A Simpler Normalization and Why It Fails}\label{secNormfail}

It is natural to ask whether the rank sum $N(S)$ in Definition \ref{aardef} is doing essential work, or whether the more familiar normalization by menu size alone would serve just as well. Consider the rule
\[
p^\circ(x,S,A):=\frac{1}{|S|\,|A|}\sum_{a\in A}\gamma(a)\,V(x,a,S),\qquad\gamma(a)>0,
\]
which differs from the Additive Attention Rule only in dividing by $|S|$ rather than by $N(S)=|S|(|S|+1)/2$. This is, in an obvious sense, the simpler of the two rules: no auxiliary rank-sum quantity needs to be computed, and $p^\circ$ inherits both defining features of the Additive Attention Rule --- averaging across attributes, and crediting an alternative by its full rank rather than only by whether it is on top --- at the lightest possible bookkeeping cost.

The rule $p^\circ$ admits essentially the same axiomatic treatment. It retains IS, BD, ST, and AVG unchanged, and satisfies natural counterparts of MA and BC in which $N(S)$ is everywhere replaced by $|S|$ and BC's $\tfrac23$-$\tfrac13$ exchange rate is replaced by a full-half split:
\begin{itemize}
    \item Counterpart of MA: \[
|S\cup T|\,p^\circ(x,S\cup T,\{a\})=|S|\,p^\circ(x,S,\{a\})+|T|\,p^\circ(x,T,\{a\})-p^\circ(x,\{x\},\{a\}),
\]
whenever $S\cap T=\{x\}$.
\item  Counterpart of BC:\begin{align*}
p^\circ(x,\{x,y\},\{a\})&=p^\circ(x,\{x\},\{a\}),\qquad p^\circ(y,\{x,y\},\{a\})=\tfrac12\,p^\circ(y,\{y\},\{a\})\\
&\text{whenever }p^\circ(x,\{x,y\},\{a\})\ge p^\circ(y,\{x,y\},\{a\}).
\end{align*}
\end{itemize}

One can check --- we omit the proof, which follows the induction used to prove Theorem \ref{thmAAR} verbatim, with $|S|$ in place of $N(S)$ and the full-half split in place of BC's $\tfrac23$-$\tfrac13$ split at each step --- that $p^\circ$ satisfies AVG, BD, IS, ST, and these two conditions if and only if it takes the displayed form. The rule also reproduces the same behavioural results as the Additive Attention Rule. With $x\succ_az\succ_ay$ and $y\succ_bx\succ_bz$ as in Section \ref{secATTAAR}, the same calculation with $|S|$ in place of $N(S)$ shows that the choice ranking of $x$ and $y$ reverses upon adding the decoy $z$ exactly when $\gamma(b)/2<\gamma(a)<\gamma(b)$ --- the identical condition obtained for the Additive Attention Rule. And with $x\succ_az\succ_ay$, $y\succ_bz\succ_bx$ as in Section \ref{secAARcomp}, the same argument shows that $p^\circ(z,S,A)>p^\circ(x,S,A)$ and $p^\circ(z,S,A)>p^\circ(y,S,A)$ require the mutually exclusive conditions $\gamma(b)>\gamma(a)$ and $\gamma(a)>\gamma(b)$, so $p^\circ$ cannot exhibit the Compromise Effect either.

Despite this, $p^\circ$ is not a genuine stochastic choice rule. Since $\sum_{x\in S}V(x,a,S)=N(S)=|S|(|S|+1)/2$ for every attribute $a$, and by rearranging terms, we get,
\[
\sum_{x\in S}p^\circ(x,S,A)=\frac{1}{|S||A|}\sum_{a\in A}\gamma(a)\cdot\frac{|S|(|S|+1)}{2}=\frac{|S|+1}{2|A|}\sum_{a\in A}\gamma(a),
\]
which grows without bound in $|S|$ for any fixed $A$ and fixed positive weights, and so exceeds one on a large enough menu, whatever the weights are. For a concrete illustration, take the Attraction-Effect configuration above with $\gamma(a)=0.4$, $\gamma(b)=0.5$ (so $\gamma(b)/2<\gamma(a)<\gamma(b)$, as the reversal requires). On $S=\{x,y,z\}$ the total is $\frac44(0.9)=0.9\le1$, giving no hint of a problem; but simply adding a fourth alternative $w$ to the menu (with $\succeq_a,\succeq_b$ extended to rank $w$ anywhere) gives
\[
\sum_{v\in\{x,y,z,w\}}p^\circ(v,\{x,y,z,w\},A)=\frac54(0.4+0.5)=1.125>1,
\]
so that $p^\circ(x^*,\{x,y,z,w\},A)=1-1.125=-0.125<0$ --- not a valid probability. The rank sum $N(S)$ in Definition \ref{aardef} is exactly what repairs this, at the cost of the slightly heavier bookkeeping described above.

\section{Conclusion}\label{sec5}

We provide a model of individual stochastic choice where a decision maker is attention-biased towards different attributes where the attributes are considered with independent probabilities. For future work, it would be interesting to explore a framework where the attention probabilities are dependent. For example, it is possible that if attention paid to an attribute is high then the attention paid to another complementary attribute is also high. This can happen when two or more attributes are related, for example, while buying a box of breakfast cereal, the `the image of the company'(if the company markets itself as a health conscious company or if the company focuses on the taste of the product) as an attribute may be linked to another attribute, `the sugar content' of the cereal. 

Beyond the multiplicative rule, we have isolated an additive variant---the Additive Attention Rule---that trades product aggregation for averaging and tops-only weighting for rank weighting. This single change lets the rule reproduce the Attraction Effect with a \emph{fixed} set of attributes, whereas MAR requires the relevant attributes to vary; the same additive linearity, however, precludes the Compromise Effect. Characterizing the choice rules that accommodate both effects within an attribute-based model of limited attention---plausibly through a non-linear aggregation of ranks---is a natural next step.
\section*{Acknowledgements and funding}
We are grateful to Bhavook Bhardwaj, Aditya Chowdhury, Sambuddha Ghosh, Ojasvi Khare, Puja Mukherjee,
Saptarshi Mukherjee, Sayan Mukherjee, Santanu Roy, Nirmalee Sahu, and
Ruhi Sonal for helpful comments. No additional funding was granted for this research.
\section*{Conflict of Interests Statement}
We do not have any conflict of interests to report.

\section*{Declaration of generative AI and AI-assisted technologies in the writing process}

During the preparation of this work the authors used Claude for
copy-editing and language refinement. After using these tools, the authors reviewed and edited the
content as needed and take full responsibility for the content of the published
article.

\bibliography{references.bib}

@article{Fishburn1973,
  author  = {Fishburn, Peter C.},
  title   = {Binary Choice Probabilities: On the Varieties of Stochastic Transitivity},
  journal = {Journal of Mathematical Psychology},
  year    = {1973},
  volume  = {10},
  number  = {4},
  pages   = {327--352}
}

@article{aguiar2019adaptive,
  author  = {Aguiar, Victor H. and Kimya, Mert},
  title   = {Adaptive Stochastic Search},
  journal = {Journal of Mathematical Economics},
  year    = {2019},
  volume  = {81},
  pages   = {74--83}
}

@article{ahumada2018luce,
  author  = {Ahumada, Alonso and {\"U}lk{\"u}, Levent},
  title   = {Luce Rule with Limited Consideration},
  journal = {Mathematical Social Sciences},
  year    = {2018},
  volume  = {93},
  pages   = {52--56}
}

@article{arrow1959rational,
  author  = {Arrow, Kenneth J.},
  title   = {Rational Choice Functions and Orderings},
  journal = {Economica},
  year    = {1959},
  volume  = {26},
  number  = {102},
  pages   = {121--127}
}

@incollection{barbera2004ranking,
  author    = {Barber{\`a}, Salvador and Bossert, Walter and Pattanaik, Prasanta K.},
  title     = {Ranking Sets of Objects},
  booktitle = {Handbook of Utility Theory},
  editor    = {Barber{\`a}, Salvador and Hammond, Peter J. and Seidl, Christian},
  volume    = {2},
  pages     = {893--977},
  publisher = {Kluwer Academic Publishers},
  address   = {Boston},
  year      = {2004}
}

@article{bergus1995clinical,
  author  = {Bergus, George R. and Chapman, Gretchen B. and Gjerde, Craig and Elstein, Arthur S.},
  title   = {Clinical Reasoning about New Symptoms Despite Preexisting Disease: Sources of Error and Order Effects},
  journal = {Family Medicine},
  year    = {1995},
  volume  = {27},
  number  = {5},
  pages   = {314--320}
}

@article{bhattacharya2021frame,
  author  = {Bhattacharya, Mihir and Mukherjee, Saptarshi and Sonal, Ruhi},
  title   = {Frame-Based Stochastic Choice Rule},
  journal = {Journal of Mathematical Economics},
  year    = {2021},
  volume  = {97},
  pages   = {102553}
}

@article{bialkova2010determines,
  author  = {Bialkova, Svetlana and van Trijp, Hans},
  title   = {What Determines Consumer Attention to Nutrition Labels?},
  journal = {Food Quality and Preference},
  year    = {2010},
  volume  = {21},
  number  = {8},
  pages   = {1042--1051}
}

@article{bordalo2012salience,
  author  = {Bordalo, Pedro and Gennaioli, Nicola and Shleifer, Andrei},
  title   = {Salience Theory of Choice under Risk},
  journal = {The Quarterly Journal of Economics},
  year    = {2012},
  volume  = {127},
  number  = {3},
  pages   = {1243--1285}
}

@article{bordalo2013salience,
  author  = {Bordalo, Pedro and Gennaioli, Nicola and Shleifer, Andrei},
  title   = {Salience and Consumer Choice},
  journal = {Journal of Political Economy},
  year    = {2013},
  volume  = {121},
  number  = {5},
  pages   = {803--843}
}

@article{bordalo2020memory,
  author  = {Bordalo, Pedro and Gennaioli, Nicola and Shleifer, Andrei},
  title   = {Memory, Attention, and Choice},
  journal = {The Quarterly Journal of Economics},
  year    = {2020},
  volume  = {135},
  number  = {3},
  pages   = {1399--1442}
}

@article{cattaneo2020random,
  author  = {Cattaneo, Matias D. and Ma, Xinwei and Masatlioglu, Yusufcan and Suleymanov, Elchin},
  title   = {A Random Attention Model},
  journal = {Journal of Political Economy},
  year    = {2020},
  volume  = {128},
  number  = {7},
  pages   = {2796--2836}
}

@incollection{chapman2000cognitive,
  author    = {Chapman, Gretchen B. and Elstein, Arthur S.},
  title     = {Cognitive Processes and Biases in Medical Decision Making},
  booktitle = {Decision Making in Health Care: Theory, Psychology, and Applications},
  editor    = {Chapman, Gretchen B. and Sonnenberg, Frank A.},
  pages     = {183--210},
  publisher = {Cambridge University Press},
  address   = {Cambridge},
  year      = {2000}
}

@article{cherepanov2013rationalization,
  author  = {Cherepanov, Vadim and Feddersen, Timothy and Sandroni, Alvaro},
  title   = {Rationalization},
  journal = {Theoretical Economics},
  year    = {2013},
  volume  = {8},
  number  = {3},
  pages   = {775--800}
}

@article{chernoff1954rational,
  author  = {Chernoff, Herman},
  title   = {Rational Selection of Decision Functions},
  journal = {Econometrica},
  year    = {1954},
  volume  = {22},
  number  = {4},
  pages   = {422--443}
}

@article{cunnington1997effect,
  author  = {Cunnington, J. P. W. and Turnbull, J. M. and Regehr, G. and Marriott, M. and Norman, G. R.},
  title   = {The Effect of Presentation Order in Clinical Decision Making},
  journal = {Academic Medicine},
  year    = {1997},
  volume  = {72},
  number  = {10},
  pages   = {S40--S42}
}

@article{dhar1997consumer,
  author  = {Dhar, Ravi},
  title   = {Consumer Preference for a No-Choice Option},
  journal = {Journal of Consumer Research},
  year    = {1997},
  volume  = {24},
  number  = {2},
  pages   = {215--231}
}

@article{fudenberg2015stochastic,
  author  = {Fudenberg, Drew and Iijima, Ryota and Strzalecki, Tomasz},
  title   = {Stochastic Choice and Revealed Perturbed Utility},
  journal = {Econometrica},
  year    = {2015},
  volume  = {83},
  number  = {6},
  pages   = {2371--2409}
}

@unpublished{galperti2024listing,
  author = {Galperti, Simone},
  title  = {Listing Specs: The Effect of Framing Attributes on Choice},
  note   = {Working paper},
  year   = {2024}
}

@article{gerasimou2018indecisiveness,
  author  = {Gerasimou, Georgios},
  title   = {Indecisiveness, Undesirability and Overload Revealed through Rational Choice Deferral},
  journal = {The Economic Journal},
  year    = {2018},
  volume  = {128},
  number  = {614},
  pages   = {2450--2479}
}

@article{gul2014random,
  author  = {Gul, Faruk and Natenzon, Paulo and Pesendorfer, Wolfgang},
  title   = {Random Choice as Behavioral Optimization},
  journal = {Econometrica},
  year    = {2014},
  volume  = {82},
  number  = {5},
  pages   = {1873--1912}
}

@article{hara2019coalitional,
  author  = {Hara, Kazuhiro and Ok, Efe A. and Riella, Gil},
  title   = {Coalitional Expected Multi-Utility Theory},
  journal = {Econometrica},
  year    = {2019},
  volume  = {87},
  number  = {3},
  pages   = {933--980}
}

@article{honda2021categorical,
  author  = {Honda, Edward},
  title   = {Categorical Consideration and Perception Complementarity},
  journal = {Economic Theory},
  year    = {2021},
  volume  = {71},
  number  = {2},
  pages   = {693--716}
}

@article{huber1982adding,
  author  = {Huber, Joel and Payne, John W. and Puto, Christopher},
  title   = {Adding Asymmetrically Dominated Alternatives: Violations of Regularity and the Similarity Hypothesis},
  journal = {Journal of Consumer Research},
  year    = {1982},
  volume  = {9},
  number  = {1},
  pages   = {90--98}
}

@article{iyengar2000choice,
  author  = {Iyengar, Sheena S. and Lepper, Mark R.},
  title   = {When Choice Is Demotivating: Can One Desire Too Much of a Good Thing?},
  journal = {Journal of Personality and Social Psychology},
  year    = {2000},
  volume  = {79},
  number  = {6},
  pages   = {995--1006}
}

@article{kalai2002rationalizing,
  author  = {Kalai, Gil and Rubinstein, Ariel and Spiegler, Ran},
  title   = {Rationalizing Choice Functions by Multiple Rationales},
  journal = {Econometrica},
  year    = {2002},
  volume  = {70},
  number  = {6},
  pages   = {2481--2488}
}

@article{kimya2018choice,
  author  = {Kimya, Mert},
  title   = {Choice, Consideration Sets, and Attribute Filters},
  journal = {American Economic Journal: Microeconomics},
  year    = {2018},
  volume  = {10},
  number  = {4},
  pages   = {223--247}
}

@article{kovach2022focal,
  author  = {Kovach, Matthew and Tserenjigmid, Gerelt},
  title   = {The Focal Luce Model},
  journal = {American Economic Journal: Microeconomics},
  year    = {2022},
  volume  = {14},
  number  = {3},
  pages   = {378--413}
}

@article{kovach2023reference,
  author  = {Kovach, Matthew and Suleymanov, Elchin},
  title   = {Reference Dependence and Random Attention},
  journal = {Journal of Economic Behavior \& Organization},
  year    = {2023},
  volume  = {215},
  pages   = {421--441}
}

@article{lancaster1966new,
  author  = {Lancaster, Kelvin J.},
  title   = {A New Approach to Consumer Theory},
  journal = {Journal of Political Economy},
  year    = {1966},
  volume  = {74},
  number  = {2},
  pages   = {132--157}
}

@article{lleras2017more,
  author  = {Lleras, Juan Sebasti{\'a}n and Masatlioglu, Yusufcan and Nakajima, Daisuke and Ozbay, Erkut Y.},
  title   = {When More Is Less: Limited Consideration},
  journal = {Journal of Economic Theory},
  year    = {2017},
  volume  = {170},
  pages   = {70--85}
}

@book{luce,
  author    = {Luce, R. Duncan},
  title     = {Individual Choice Behavior: A Theoretical Analysis},
  publisher = {Wiley},
  address   = {New York},
  year      = {1959}
}

@article{manzini2007sequentially,
  author  = {Manzini, Paola and Mariotti, Marco},
  title   = {Sequentially Rationalizable Choice},
  journal = {American Economic Review},
  year    = {2007},
  volume  = {97},
  number  = {5},
  pages   = {1824--1839}
}

@article{manzini2012categorize,
  author  = {Manzini, Paola and Mariotti, Marco},
  title   = {Categorize Then Choose: Boundedly Rational Choice and Welfare},
  journal = {Journal of the European Economic Association},
  year    = {2012},
  volume  = {10},
  number  = {5},
  pages   = {1141--1165}
}

@article{manzini2014stochastic,
  author  = {Manzini, Paola and Mariotti, Marco},
  title   = {Stochastic Choice and Consideration Sets},
  journal = {Econometrica},
  year    = {2014},
  volume  = {82},
  number  = {3},
  pages   = {1153--1176}
}

@article{masatlioglu2012revealed,
  author  = {Masatlioglu, Yusufcan and Nakajima, Daisuke and Ozbay, Erkut Y.},
  title   = {Revealed Attention},
  journal = {American Economic Review},
  year    = {2012},
  volume  = {102},
  number  = {5},
  pages   = {2183--2205}
}

@article{masatlioglusqb,
  author  = {Masatlioglu, Yusufcan and Ok, Efe A.},
  title   = {Rational Choice with Status Quo Bias},
  journal = {Journal of Economic Theory},
  year    = {2005},
  volume  = {121},
  number  = {1},
  pages   = {1--29}
}

@incollection{mcfadden1990stochastic,
  author    = {McFadden, Daniel and Richter, Marcel K.},
  title     = {Stochastic Rationality and Revealed Stochastic Preference},
  booktitle = {Preferences, Uncertainty, and Optimality: Essays in Honor of Leo Hurwicz},
  editor    = {Chipman, John S. and McFadden, Daniel and Richter, Marcel K.},
  pages     = {161--186},
  publisher = {Westview Press},
  address   = {Boulder, CO},
  year      = {1990}
}

@article{ok2002utility,
  author  = {Ok, Efe A.},
  title   = {Utility Representation of an Incomplete Preference Relation},
  journal = {Journal of Economic Theory},
  year    = {2002},
  volume  = {104},
  number  = {2},
  pages   = {429--449}
}

@article{plott,
  author  = {Plott, Charles R.},
  title   = {Path Independence, Rationality, and Social Choice},
  journal = {Econometrica},
  year    = {1973},
  volume  = {41},
  number  = {6},
  pages   = {1075--1091}
}

@article{puppe1996axiomatic,
  author  = {Puppe, Clemens},
  title   = {An Axiomatic Approach to ``Preference for Freedom of Choice''},
  journal = {Journal of Economic Theory},
  year    = {1996},
  volume  = {68},
  number  = {1},
  pages   = {174--199}
}

@article{salant2008f,
  author  = {Salant, Yuval and Rubinstein, Ariel},
  title   = {{(A, f)}: Choice with Frames},
  journal = {The Review of Economic Studies},
  year    = {2008},
  volume  = {75},
  number  = {4},
  pages   = {1287--1296}
}

@article{simonson1989choice,
  author  = {Simonson, Itamar},
  title   = {Choice Based on Reasons: The Case of Attraction and Compromise Effects},
  journal = {Journal of Consumer Research},
  year    = {1989},
  volume  = {16},
  number  = {2},
  pages   = {158--174}
}

\section{Appendix A}\label{sec7}
\subsection{Pareto Dominance and MAR: Some Motivational Examples}\label{motivation}

\par As a justification for why MAR is a behaviourally plausible rule, we connect this to a very well-known choice procedure in the choice theory literature given our context: \textit{Pareto dominance}. Pareto dominance states that an alternative $x$ is chosen if and only if $x$ is preferred to all other alternatives in the feasible menu of alternatives and in all relevant attributes; if such an $x$ exists, it is said to be a Pareto dominant alternative in that choice problem\footnote{Formally, a choice problem is an ordered pair, $(S,A)$, where $S$ is the menu of alternatives and $A$ is the set of relevant attributes.}. If such an $x$ does not exist, there is no unique top alternative and the DM then prefers to exercise choice deferral (the act of currently not choosing anything from the feasible menu of alternatives or sticking to the default alternative). \cite{ok2002utility} makes a case for this. For experimental evidence on choice deferral with multidimensional alternatives (the dimensions could be attributes or states of the world) when no single alternative has a decisive advantage, see \cite{dhar1997consumer}. 

As a natural example where Pareto dominance is useful, consider the case where a family is making an expensive purchase like that of a house. Usually each family member's opinion (the opinions of the different family members can be thought of as attributes) is considered to be important for making a decision.

 Pareto dominance captures the idea of unanimity at the top in such cases: a house is bought if and only if each family member has that house at the top of her/his/their preference. So, if there is no such alternative, the family sticks to their default alternative of living in their existing house. Hence, each family member has a veto power. For MAR, each attribute is important (in terms of attention, as each attribute is more likely to be considered than not, and also in terms of preference,  as each attribute has a `veto power', hence if there is no such Pareto dominant alternative, the DM gets conflicted and sticks to the default. However, MAR also allows the DM to choose feasible alternatives from the menu with non-zero probability even in the absence of a Pareto dominant alternative (see Example \ref{egA}). This is due to the limited attention paid to each attribute. However, even in such cases the default alternative gets the highest choice probability.  We show the following in Appendix \ref{whenindoubt}: If there is no Pareto  dominant alternative, then, according to MAR, the default alternative gets chosen more often than not (with probability greater than $\dfrac{1}{2}$), and hence gets the highest choice probability\footnote{For now, consider the the simple case of two-alternatives $\{x,y\}$ and two-attributes $\{a,b\}$ with $x \succ_ay$ and $y \succ_bx$. The sum $p^S(x, \{a,b\})+p^S(y,\{a,b\})=\gamma(a).(1-\gamma(b))+\gamma(b).(1-\gamma(a))$ is maximized at the boundary points $(\gamma(a)=0.5,\gamma(b)=0.5)$, $(\gamma(a)=1,\gamma(b)=0.5)$, and $(\gamma(a)=0.5,\gamma(b)=1)$ with the maximized value of the sum being $0.5$. However, $\gamma(a),\gamma(b)\in (0.5,1]$. Hence, for any value of $\gamma(a),\gamma(b)\in (0.5,1]$, the value of the sum $\gamma(a).(1-\gamma(b))+\gamma(b).(1-\gamma(a))$ is less than $0.5$. Thus, $p^S(x, \{a,b\})+p^S(y,\{a,b\})<\dfrac{1}{2}$. Hence, $p^S(x^*,\{x,y\},\{a,b\})>\dfrac{1}{2}$. This observation is in agreement with the Pareto dominance principle: if there is no Pareto dominant alternative, then the DM sticks to her default. In the stochastic case, this is weakened to the point that the DM in such cases sticks to her default more often than not, that is, with probability greater than $\dfrac{1}{2}$. The proof of the general case is given in Appendix \ref{whenindoubt}. \label{twoatt}} (however, the converse is not true, see Remark \ref{rempareto1} for more details). Also, the default alternative, in addition to capturing residual choice probabilities (as is true in any of the stochastic choice rules which are in multiplicative/product form), serves a behavioural purpose of sticking to a status quo: where if a Pareto dominant alternative does not exist or there are too many attributes to focus on then a DM would want to stick to the default alternative, i.e., the status quo. However, this differs from the notion of status quo bias as in \cite{masatlioglusqb} where the status quo alternative is actively compared with the other alternatives in the menu of feasible alternatives. For an illuminating discussion on default alternatives and empty choice sets, see \cite{gerasimou2018indecisiveness} and the references therein.

\par We formally show in Appendix \ref{sec4.3} that given full attention (each attribute is considered with probability 1), a deterministic choice function based on MAR satisfies Pareto dominance.

We now provide a few stylized examples that highlight when an alternative gets the highest choice probability in a choice problem and when the default alternative gets the highest choice probability in a choice problem with a Pareto dominant alternative\footnote{The case without a Pareto dominant alternative has been covered in Example \ref{egA}.}: 

\begin{example}\rm{}\label{egB} Suppose $S=\{x,y,z\}$ and $A=\{a_{1},a_{2}\}$. Let the two orderings with respect to the attributes $a_{1}$ and $a_{2}$ be $x\succ_{a_{1}}y\succ_{a_{1}}z$ and $x\succ_{a_{2}}z\succ_{a_{2}}y$ respectively. Let $\gamma(a_{1})=0.9$ and $\gamma(a_{2})=0.85$ be the two attention probabilities of the attributes. The choice probabilities of the three alternatives from the set $S=\{x,y,z\}$ according to MAR, $p^{S}$ are as follows: 

\begin{itemize}
    \item [(i)] $p^{S}(x,S,A)= \gamma(a_{1})\gamma(a_2)=(0.9).(0.85)=0.765$.

    \item [(ii)] $p^{S}(y,S,A)= \left(\dfrac{1-\gamma(a_{1})}{|S|-1}\right)\left(\dfrac{1-\gamma(a_{2})}{|S|-1}\right)=\left(\dfrac{1-0.9}{3-1}\right)\left(\dfrac{1-0.85}{3-1}\right)=0.05(0.075)=0.00375$

     \item [(iii)] $p^{S}(z,S,A)=\left(\dfrac{1-\gamma(a_{1})}{|S|-1}\right)\left(\dfrac{1-\gamma(a_{2})}{|S|-1}\right)=\left(\dfrac{1-0.9}{3-1}\right)\left(\dfrac{1-0.85}{3-1}\right)=0.05(0.075)=0.00375$.
\end{itemize}

\vspace{0.12in}
Here $x$ is a Pareto dominant alternative in $(S,A)$ which gets the highest choice probability in comparison with the rest of the alternatives in $S$ ($x$ procures a choice probability that is 204 times that of $y$ and $z$). Interestingly, $p^S(x,S,A)=0.765>0.2275=p^S(x^*,S,A)$. This is because $x$ is Pareto dominant \textit{and} there are \textit{few} relevant attributes and relatively \textit{high} attention probabilities of the two attributes. We show in the next two examples if either of the above two are absent then even in presence of a Pareto dominant alternative, the default alternative gets the highest choice probability. 
\end{example}

\begin{example}\rm{}\label{eg4} Suppose $S=\{x,y,z\}$ and $A=\{a_{1},a_{2}\}$. Let the two orderings with respect to the attributes $a_{1}$ and $a_{2}$ be $x\succ_{a_{1}}y\succ_{a_{1}}z$ and $x\succ_{a_{2}}z\succ_{a_{2}}y$ respectively. Let $\gamma(a_{1})=0.6$ and $\gamma(a_{2})=0.6$ be the two attention probabilities of the attributes. The choice probabilities of the three alternatives from the set $S=\{x,y,z\}$ according to MAR, $p^{S}$ are as follows: 

\begin{itemize}
    \item [(i)] $p^{S}(x,S,A)= \gamma(a_{1})\gamma(a_2)=(0.6).(0.6)=0.36$.

    \item [(ii)] $p^{S}(y,S,A)= \left(\dfrac{1-\gamma(a_{1})}{|S|-1}\right)\left(\dfrac{1-\gamma(a_{2})}{|S|-1}\right)=\left(\dfrac{1-0.6}{3-1}\right)\left(\dfrac{1-0.6}{3-1}\right)=0.2(0.2)=0.04$

     \item [(iii)] $p^{S}(z,S,A)=\left(\dfrac{1-\gamma(a_{1})}{|S|-1}\right)\left(\dfrac{1-\gamma(a_{2})}{|S|-1}\right)=\left(\dfrac{1-0.6}{3-1}\right)\left(\dfrac{1-0.6}{3-1}\right)=0.2(0.2)=0.04$.
\end{itemize}

\vspace{0.12in}
Here $x$ is a Pareto dominant alternative in $(S,A)$ which gets the highest choice probability in comparison with the rest of the alternatives in $S$ ($x$ procures a choice probability that is 9 times that of $y$ and $z$). However, now $p^S(x,S,A)=0.36<0.56=p^S(x^*,S,A)$. This is due to somewhat lower consideration probabilities of the two attributes in comparison with the previous example. Thus, even in presence of a Pareto dominant element, it is possible that the default alternative gets a higher choice probability.

\end{example}
\begin{example}\rm{}\label{egC} Suppose $S=\{x,y,z\}$ and $A=\{a_{1},a_{2}, a_3,a_4,a_5\}$. Let the five orderings with respect to the attributes $a_{1},a_2,a_3,a_4, \text{ and }a_5$ be $x\succ_{a_{1}}y\succ_{a_{1}}z$, $x\succ_{a_{2}}z\succ_{a_{2}}y$, $x\succ_{a_{3}}y\succ_{a_{3}}z$, $x\succ_{a_{4}}z\succ_{a_{4}}y$, and $x\succ_{a_{5}}y\succ_{a_{5}}z$, respectively. Let $\gamma(a_{1})=0.9$, $\gamma(a_{2})=0.85$, $\gamma(a_{3})=0.85$, $\gamma(a_{4})=0.85$, and $\gamma(a_{5})=0.85$ be the five attention probabilities of the attributes. The choice probabilities of the three alternatives from the set $S=\{x,y,z\}$ according to MAR, $p^{S}$ are as follows: 

\begin{itemize}
    \item [(i)] $p^{S}(x,S,A)= \gamma(a_{1})\gamma(a_2)\gamma(a_3)\gamma(a_4)\gamma(a_5)=(0.9).(0.85).(0.85).(0.85).(0.85)=0.4698$.

    \item [(ii)] $p^{S}(y,S,A)= \left(\dfrac{1-\gamma(a_{1})}{|S|-1}\right)\left(\dfrac{1-\gamma(a_{2})}{|S|-1}\right)\left(\dfrac{1-\gamma(a_{3})}{|S|-1}\right)\left(\dfrac{1-\gamma(a_{4})}{|S|-1}\right)\left(\dfrac{1-\gamma(a_{5})}{|S|-1}\right)=\left(\dfrac{1-0.9}{3-1}\right)\left(\dfrac{1-0.85}{3-1}\right)^4=0.05(0.075)^4=0.000002$

     \item [(iii)] $p^{S}(z,S,A)=\left(\dfrac{1-\gamma(a_{1})}{|S|-1}\right)\left(\dfrac{1-\gamma(a_{2})}{|S|-1}\right)\left(\dfrac{1-\gamma(a_{3})}{|S|-1}\right)\left(\dfrac{1-\gamma(a_{4})}{|S|-1}\right)\left(\dfrac{1-\gamma(a_{5})}{|S|-1}\right)=\left(\dfrac{1-0.9}{3-1}\right)\left(\dfrac{1-0.85}{3-1}\right)^4=0.05(0.075)^4=0.000002$.
\end{itemize}

\vspace{0.12in}
Here $x$ is a Pareto dominant alternative in $(S,A)$ which gets the highest choice probability in comparison with the rest of the alternatives in $S$ ($x$ procures a choice probability that is 234900 times that of $y$ and $z$). However, $p^S(x,S,A)=0.4698<0.530196=p^S(x^*,S,A)$. In spite of relatively higher attention probabilities of attributes and the presence of a Pareto dominant element, the default alternative gets the highest choice probability. Here, due to a large number of relevant attributes the DM gets overwhelmed and is more likely to choose a default alternative. This shows that \textit{both} limited attention to attributes and high number of attributes can generate choice overload in our model even in presence of a Pareto dominant alternative. However, if we let $\gamma(a_1)=\gamma(a_2)=\gamma(a_3)=\gamma(a_4)=\gamma(a_5)=0.9$, we get $p^S(x,S,A)=0.59049$, and the Pareto dominant element $x$ now gets the highest choice probability ($x$ procures a choice probability that is 1889568 times that of $y$ and $z$), in particular higher than that of the default alternative.

\begin{remark}\label{rempareto1}
    Thus, in the above example we note that even in presence of a Pareto dominant alternative, if there are too many attributes or if the attention parameters are not high enough, the default alternative can be chosen with the highest probability. This is because the large number of attributes itself exerts high cognitive load on the DM and she resorts to her default alternative more often than not. Also, relatively less attention being paid to attributes when the number of attributes is not that high can also result in the default alternative being chosen more often than not (as in Example \ref{eg4}). Formally, it can be shown that a sufficient condition for the default alternative to have the highest choice probability, \textit{even in the presence of a Pareto dominant alternative}, is: \[\prod_{a\in A}\gamma(a)+\bigg(\dfrac{1}{|S|-1}\bigg)^{|A|-1}.\prod_{a\in A}(1-\gamma(a)) <\dfrac{1}{2},\] where there menu of feasible alternatives is $S$ and the menu of relevant attributes is $A$.
\end{remark}
 Thus, MAR is capable of explaining a wide array of choice behaviours. A high choice probability of the default alternative is also motivated by the fact that a DM maybe bombarded by advertisements related to various products all over the Internet during a browsing session and in many cases she does not buy anything, i.e., she sticks to her default. This happens more so when there are too many attributes to focus on and/or the consideration probability of the attributes are not high enough. Also, DMs may choose to window shop (the act of just browsing various products) more often than she actually buys a product\footnote{This can be related to the DM wanting to browse products online while procrastinating.}. For advertising purposes, a firm should produce a Pareto dominant element and then the firm can negotiate, with the platform (like Amazon) which controls what filters a consumer is allowed to focus on, i.e., the relevant attributes, to lower the relevant number of attributes and advertise well enough on the low number of allowed attributes so that the consideration probability of the attributes will be high enough. This would be an ironclad way to ensure that the alternative produced by the firm will get a high choice probability. 
    
\end{example}

We provide one final example to show that our rule is robust (in terms of relative choice probabilities) to minor perturbations in the presence of a large number of attributes:
\begin{example}\rm{}\label{egRob}
Suppose that $S=\{x,y\}$ and $x \succ_a y, \forall a \in A \text{ such that } |A|=99$. Let $\gamma(a)=0.9, \forall a \in A$. Then, according to MAR, $p^S(x,S,A)=(0.9)^{99}=0.0000295$ and $p^S(y,S,A)=(\dfrac{1-0.9}{2-1})^{99}\approx 0$. Hence, $x$ is chosen with a significantly higher probability than $y$ in $(S,A)$. Now, if an attribute $b\not \in A$ is added to the mix such that $y\succ_b x$, and $\gamma(b)=0.9$, then according to MAR, $p^S(x,S,A\cup \{b\})=(0.9)^{99}.(\dfrac{1-0.9}{2-1})=0.00000295$ and $p^S(y,S,A\cup \{b\})=(0.9).(\dfrac{1-0.9}{2-1})^{99}\approx 0$. Thus, even in $(S,A \cup \{b\})$, $x$ is chosen with a significantly higher probability than $y$. Therefore, even though the addition of $b$ to $A$ lowers the choice probability of $x$ in absolute terms (and hence increases the choice probability of the default alternative), the relative choice probabilities (of $x$ and $y$) are still robust. 
\end{example}

\begin{remark}
    As the above examples show, when a Pareto dominant alternative exists, even when the choice probability of the default alternative is the highest, the Pareto dominant alternative has a significantly higher choice probability than the other feasible alternatives (see Example \ref{egC}, in particular). Thus, MAR can incorporate choice deferral and limited attention, and at the same time, it respects unanimity at the top of the attributes-based preferences.
    \end{remark}
    \subsection{When in Doubt, Probably Stay out: The High Probability of the Default Alternative in MAR}\label{whenindoubt}
In this appendix, we show that if there is no Pareto dominant alternative in a choice problem, then according to MAR, the default alternative is chosen more often than not. As discussed in the Introduction and in Appendix \ref{motivation}, in MAR, each attribute is important and if there is no unique top alternative  (i.e, a Pareto dominant alternative, which is preferred to all the other feasible alternatives in a menu via all the relevant attributes), the DM gets conflicted and sticks to her default/status quo alternative more often than not. This is in agreement with the deterministic choice procedure (see Appendix \ref{sec4.3} for more details) associated with Pareto dominance (where the DM \textit{always} sticks to her default alternative in such cases, with probability equal to $1$) and in MAR, this is weakened to the DM sticking to her default alternative \textit{more often than not}, that is, with probability greater than $\dfrac{1}{2}$.\footnote{Even in absence of a Pareto dominant alternative, some alternative maybe chosen with a significantly higher probability than the other feasible alternatives (but less than $\dfrac{1}{2}$, hence less than the choice probability of the default alternative). For example, let the set of alternatives be $S=\{x,y\}$ and the set of attributes be $A=\{a,b\}$. Let $x \succ_a y, y\succ_bx$ and $\gamma(a)=0.95, \gamma(b)=0.55$. Then according to MAR, $p(x,S,A)=0.95\times0.45=0.4275$, $p(y,S,A)=0.55\times 0.05=0.025$, and $p(x^*,S,A)=1-0.4275-0.0275=0.545$. Thus, even though $x$ is not a Pareto dominant element, it still procures a moderately high choice probability and while $y$ gets a very low choice probability. This is because the attribute in which $x$ is at the top grabs relatively high attention ($\gamma(a)=0.95$) while the attribute in which $y$ is at the top grabs relatively low attention ($\gamma(b)=0.55$).} The above is consistent with the experimental evidence provided by \cite{dhar1997consumer}, where the DM prefers to exercise choice deferral when no single alternative has a decisive advantage over the other alternatives. A stylized way of saying this is: \textit{when in doubt, probably stay out (stick to your default alternative more often than not).}
\begin{proposition}
    Let $p$ satisfy MAR. If in a choice problem $(S,A)$, there is no Pareto dominant alternative, then $p(x^*,S,A)>\dfrac{1}{2}$.
\end{proposition}
\begin{proof}
   Let $p$ satisfy MAR. Take any $S \in \mathcal{X}$ and any $A \in \mathcal{A}$. Suppose that there is no Pareto dominant alternative in the choice problem $(S,A)$. This implies for each $w \in S$, $\exists d \in A$ and $\exists r \in S\setminus\{w\}$ such that $r \succ_d w$. Also, it is obvious that $|A|\geq 2$. We need to show that $p(x^*,S,A)>\dfrac{1}{2}$ (i.e., the default alternative is chosen more often than not). Consider the case of $(S, \{a,b\})$. Suppose that $|S| \geq 3$. This is without loss of generality because the case where $|S|=2$ has already been shown in footnote \ref{twoatt} in Appendix \ref{motivation}. Now consider any two distinct attributes $a \text{ and } b$ such that $x\succeq_a z, \forall z \in S \setminus \{x\}$ and $y\succeq_b z, \forall z \in S \setminus \{y\}$, where $x \not =y$ and $x,y \in S$. Note that such $a,b,x, \text{ and }y$ exist as there is no Pareto dominant alternative in $(S,A)$. We analyse $(S,\{a,b\})$. By the definition of MAR, we get $p(x,S, \{a,b\})=\dfrac{\gamma(a)(1- \gamma(b))}{|S|-1}$, $p(y,S, \{a,b\})=\dfrac{\gamma(b)(1- \gamma(a))}{|S|-1}$, and $p(t,S, \{a,b\})=\dfrac{(1-\gamma(a))(1- \gamma(b))}{(|S|-1)^2}, \forall t \in S \setminus\{x,y\}$. We show that $\sum_{z \in S}p(z,S, \{a,b\}) <\dfrac{1}{2}$. Then, by \textit{Regularity in Attributes} (see Proposition \ref{reginatt}), we will get $\sum_{z \in S}p(z,S, A) <\dfrac{1}{2} \text{ as } \{a,b\}\subseteq A$. This, in turn, implies $p(x^*, S,A)= 1-\sum_{z \in S}p(z,S, A) >\dfrac{1}{2}$, which will establish our result. Thus, it is sufficient to show that $\sum_{z \in S}p(z,S, \{a,b\}) <\dfrac{1}{2}$. Given our observations above, we have
   \[\sum_{z \in S}p(z,S, \{a,b\}) =p(x,S, \{a,b\})+p(y,S, \{a,b\})+\sum_{t \in S \setminus\{x,y\}}p(t,S, \{a,b\})\]
   \[=\dfrac{\gamma(a)(1- \gamma(b))}{|S|-1}+\dfrac{\gamma(b)(1- \gamma(a))}{|S|-1}+(|S|-2).\dfrac{(1-\gamma(a))(1- \gamma(b))}{(|S|-1)^2}\]
   \[=\dfrac{\gamma(a)+\gamma(b)-|S|\gamma(a)\gamma(b)+|S|-2}{(|S|-1)^2}.\]
   It is straightforward to show, by using optimization techniques, that the above expression is maximized at $(\gamma(a)=\dfrac{1}{|S|},\gamma(b)=\dfrac{1}{|S|})$, $(\gamma(a)=1, \gamma(b)=\dfrac{1}{|S|})$ and $(\gamma(a)=\dfrac{1}{|S|},\gamma(b)=1)\text{ as } |S| < \infty$.\footnote{Note that the maximum we get from these values will be lesser than the actual maximum that can be attained in our model because $\dfrac{1}{|S|}\not \in (\dfrac{1}{2},1], \text{ as } |S| \geq 3$. So, it is sufficient to show that the maximum value that we get from using $(\gamma(a)=\dfrac{1}{|S|},\gamma(b)=\dfrac{1}{|S|})$, $(\gamma(a)=1, \gamma(b)=\dfrac{1}{|S|})$ and $(\gamma(a)=\dfrac{1}{|S|},\gamma(b)=1)$, is less than $ \dfrac{1}{2}$ for our purposes.} We can plug any of these values in the aforementioned expression to check for its maximum value. We plug in $(\gamma(a)=\dfrac{1}{|S|},\gamma(b)=\dfrac{1}{|S|})$, to get
   \[\max \sum_{z \in S}p(z,S, \{a,b\})=\dfrac{\dfrac{1}{|S|}+\dfrac{1}{|S|}-|S|.\dfrac{1}{|S|}.\dfrac{1}{|S|}+|S|-2}{(|S|-1)^2}\]
   \[=\dfrac{\dfrac{1}{|S|}+|S|-2}{(|S|-1)^2}=\dfrac{\dfrac{(|S|-1)^2}{|S|}}{(|S|-1)^2}=\dfrac{1}{|S|}\leq \dfrac{1}{3}<\dfrac{1}{2} \text{, as } |S| \geq 3.\] Thus, as argued above, $p(x^*,S,A) > \dfrac{1}{2}, \forall A \in \mathcal{A}, \forall S \in \mathcal{X}$, when there is no Pareto dominant alternative in $(S,A)$. This concludes the proof of our result.
\end{proof}

\subsection{Full Attention and Pareto Dominance: Deterministic Foundations of MAR}\label{sec4.3}

We now seek to establish that the deterministic variant of MAR satisfies Pareto dominance. The motivational connections between MAR and Pareto dominance have been discussed in detail in Appendix \ref{motivation}, and here we establish some formal results along those lines.\footnote{Another formal result was shown in Appendix \ref{whenindoubt}. } We first define Pareto dominance for deterministic choice functions which is borrowed from the literature on choice functions (\cite{arrow1959rational} and \cite{chernoff1954rational}).

Let $c:\mathcal{X}\times \mathcal{A}\to \mathbf X$ be a deterministic choice function such that $c(S, A)\in S$ for any $S\in\mathcal{X}$ and $A\in\mathcal{A}$. Let $\succeq_{a}$ be the preference ordering according to an attribute $a\in \mathbf A$. \begin{definition}
\textbf{Full attention} (FA):  MAR satisfies \textbf{FA} if the DM pays full attention to all the attributes, i.e., $\gamma(a)=1$ for all $a\in \mathbf{A}$.
    
\end{definition}

\begin{definition}{\textbf{Pareto Dominance:}}\label{Pareto}
   A deterministic choice function $c$ satisfies Pareto dominance if for any $S\in \mathcal{X}$ and $A\in\mathcal{A}$ the following holds:
   \begin{itemize}
   \item[\emph{(i)}] $[c(S,A)=\{x\}] \iff [x\succ_{a} y \;\;\forall\, a \in A, \forall\, y \in S\backslash\{x\}]$; and 
   \item[\emph{(ii)}] $[c(S,A)=\{x^{*}\}] \iff [\forall\, x\in S \;\exists\, a\in A , y\in S\backslash\{x\}\mbox{ such that } y\succ_{a}x]$.
\end{itemize}
\end{definition}
In agreement with what was discussed in Appendix \ref{motivation} and Appendix \ref{whenindoubt}, the above states that Pareto dominance requires that a choice is rationalizable by it if and only if the chosen alternative dominates \textit{all} the other alternatives in the menu through \textit{ all} the relevant attributes. If there is no Pareto dominant element, then the choice set is empty (or the default alternative is chosen). In terms of stochastic choice, the above translates into the fact that an alternative is chosen with probability 1 if and only if it dominates all other alternatives in the menu in all relevant attributes and given that $\gamma(a)=1$ for all $a\in A$ (by FA). If there is no Pareto dominant element, the default alternative is chosen with probability 1.

Let $c^{s}(S, A)$ denote the choice function derived from MAR as follows: (i) for any $x\in \mathbf X$, $c^{s}(S,A)=\{x\}\text{ if and only if } p^{s}(x, S, A)=1$ (ii) otherwise $c^{s}(S,A)=\{x^{*}\}$, with $p^s(x^*,S,A)=1$. Note that by the definition of MAR and FA, $p(x, S, A)=1$ if and only if $x\succ_{a}y$ for all $a\in A$ for all $y\in \mathbf X \setminus \{x\}$. Consistent with the above, for any $a\in A$ we can define the attribute orderings $\succ_{a}$ over alternatives in revealed preference terms as: $x \succ_{a} y \iff \{x\}=c^{s}(\{x,y\},\{a\}) \iff p^s(x,\{x,y\},  \{a\})=1, \forall x,y \in \mathbf{X}$.

\begin{proposition} Suppose MAR satisfies FA. The deterministic choice function $c^{s}$, based on MAR, satisfies Pareto Dominance.
\end{proposition}
\begin{proof} Take any $S\in \mathcal X$ and $A\in \mathcal A$. We first verify Condition (i) of Definition \ref{Pareto}. To this end, suppose $x = c^{s}(S,A)$. By definition of MAR, it must be that $p^{s}(x,S,A)=1$. However this is only possible if $x\succ_{a}y$ for all $a\in A$ and all $y\in S\backslash\{x\}$ and $\gamma(a)=1$ for all $a\in A$. Therefore, $A_x(S)=A$. By definition of MAR, \[p^s(x,S,A)=\prod_{a \in A}\gamma(a)= 1,\]
as $\gamma(a)=1$ for all $a\in A$. 

Similarly, one can show that \[[x\succ_{a}y\; \forall\;a\in A \; \forall\; y\in S\backslash \{x\}]\implies [p^{s}(x,S,A)=1]\implies [c^{s}(S,A)=x]\] where the first implication holds under FA.

We now verify Condition (ii) of Definition \ref{Pareto}. Suppose $c^{s}(S,A)=\{x^{*}\}$ for some $S\in\mathcal{X}$ and $A\in\mathcal{A}$. By definition of MAR, $c^{s}$ and FA, we have 
$p^s(x,S,A)=0$ for all $x\in S$. This implies that for each $x\in S$, $\exists\; a\in A$, $\exists\; y\in S\backslash\{x\}$ such that $y\succ_{a}x$.
Conversely, assume that \(\forall x \in S, \exists\; a \in A,  \exists\; y \in S \backslash \{x\}\) such that \(y \succ_{a} x\). This implies that \(A_{x}^{c}(S) \not = \emptyset, \forall x \in S\). By MAR and FA, 

\begin{align*} p^{s}(x,S,A) &=\prod_{a \in A_{x}(S)}\gamma(a)\dfrac{\prod_{a \in A^{c}_{x}(S)}(1-\gamma(a))}{(|S|-1)^{|A_{x}^{c}(S)|}}\\ &=\prod_{a \in A_{x}(S)} \gamma(a) \times \dfrac{0}{(|S|-1)^{|A_{x}^{c}(S)|}}\\
\\&=0.\end{align*}
    
By the definition of SCR, 
\[
\sum_{x \in S}p^{s}(x,S,A)+p^{s}(x^{*},S,A)=1.
\] 
Therefore, \(p^{s}(x^{*},S,A)=1\) which implies that $c^{s}(S,A)=\{x^{*}\}$. 
\end{proof}

\section{Appendix B}\label{sec6}

\subsection{Independence of the Axioms for MAR}

In this appendix, we prove the independence of the Axioms. To show this, we provide an example of an SCR for each Axiom listed in the statement of the Theorem \ref{thm1} which satisfies all Axioms except that one. 

\begin{claim} The following SCR-1 satisfies all the Axioms except \textbf{IN}.
    \begin{definition}[SCR-1] An SCR, $p^{1}$ is a SCR-1 if for all $a\in \mathbf A$ there exist an ordering $\succeq_{a}$, an attention parameter $\gamma(a)\in (\dfrac{1}{2},1]$ and an ordering $\curlyeqsucc$ on $\mathbf A$ such that for any $S\in\mathcal{X}$, $x\in S$ and $A\in \mathcal A$,   \begin{displaymath}
p^{1}(x,S,A)=\begin{cases} \gamma(a^{*}) &\quad \text{ if } a^{*}\in {A}_{x}(S) \text{ where } a^{*}=\max(\curlyeqsucc, A); \\
\dfrac{1-\gamma(a^{*})}{|S|-1} &\quad \text{ otherwise.}
    \end{cases}
    \end{displaymath}
    \end{definition}

   
   {\rm For an example, consider ${A}=\{a_{1},a_{2}\}$ with $x\succeq_{a_{1}}y\succeq_{a_{1}}z$, $y\succeq_{a_{2}}x\succeq_{a_{2}}z$, and $a_{1}\curlyeqsucc a_{2}$, then for $S=\{x,y,z\}$,
  \[p^{1}(x,S, A)=\gamma(a_{1}),\; p^{1}(y, S, A)=p^{1}(z,S, A)=\dfrac{1-\gamma(a_{1})}{2} \]
  while, 
  \[p^{1}(x,\{x,y\}, \{a_{1}\})=p^{1}(x,\{x\},\{a_{1}\})=\gamma(a_{1})\]
  and 
\[p^{1}(x,\{x,y\}, \{a_{2}\})=1-\gamma(a_{2}), \; p^{1}(y,\{y\},\{a_{2}\})=p(x,\{x\},a_{2})=\gamma(a_{2}).\]

It is easy to verify that the rule does not satisfy \text{IN} since 
\[
p^{1}(x,S,A)=\gamma(a_{1})\neq p^{1}(x,S, \{a_{1}\}) p^{1}(x,S, \{a_{2}\})=\gamma(a_{1})\left(\dfrac{1-\gamma(a_{2})}{|S|-1}\right).
\]
 The rule satisfies all the other Axioms.}
  
\end{claim}

\begin{claim} The following SCR-2 satisfies all the Axioms except \textbf{FS}.
   \begin{definition}[SCR-2] An SCR, $p^{2}$ is SCR-2 if for all $a\in \mathbf A$ there exists an ordering $\succeq_{a}$ and an attention parameters $\gamma(a)\in (\dfrac{1}{2},1]$ such that for any $S\in\mathcal{X}$, $x\in S$ and $A\in \mathcal A$,   \begin{displaymath}
p^{2}(x,S,A)=\begin{cases}
\prod_{a\in A_{x}(S)}\gamma(a)\dfrac{\prod_{b\in A_{x}^{c}(S)}(1-\gamma(b))}{2(|S|-1)^{|A_{x}^{c}(S)|}}, &\quad \text{ if }  A_{x}(S) \not = \emptyset,  A_{x}^{c}(S) \not = 
\emptyset; \\
\ \prod_{a \in A_{x}(S)}\gamma(a), & \quad \text{if } A_{x}(S)=A; \\
 \dfrac{\prod_{b \in  A_{x}^{c}(S) }(1-\gamma(b))}{2(|S|-1)^{|A_{x}^{c}(S)|}}, &\quad \text{ if } A_{x}^{c}(S)=A. \\   
    \end{cases}
    \end{displaymath}
\end{definition}
\end{claim}

Consider $S=\{x,y\}$ and $A=\{a\}$ such that $x\succ_{a}y$. This rule does not satisfy \textbf{FS} since 
\[p^{2}(x,\{x,y\},\{a\})=\gamma(a), \; p^{2}(y,\{x,y\},\{a\})=\dfrac{1-\gamma(a)}{2}, \text{ and } p^{2}(x^{*},\{x,y\},\{a\})=\dfrac{1-\gamma(a)}{2}.\]

It is easy to verify that SCR-2 satisfies all the other Axioms.

\begin{claim} The following SCR-3 satisfies all the Axioms except \textbf{IS}.
   \begin{definition}[SCR-3] An SCR, $p^{3}$ is SCR-3 if there exists an ordering $\succeq_{a}$ and attention parameters $\gamma(a)\in (\dfrac{1}{2},1]$ for all $a\in \mathbf A$ such that for any $S\in\mathcal{X}$, $x\in S$, and $A\in\mathcal{A}$,
   \begin{displaymath}
p^{3}(x,S,A)=\begin{cases}
p^{s}(x,S,A), &\quad \text{ if } |S| \geq 2; \\
\ \prod_{a \in A_{x}(\mathbf{X})}\gamma(a), & \quad \text{if } A_{x}(\mathbf{X})=A,|S|=1; \\   
\ \prod_{a \in A_{x}^{c}(\mathbf{X})}(1-\gamma(a)), &\quad \text{ if } A_{x}^{c}(\mathbf{X})=A, |S|= 1; \\ 
\prod_{a\in A_{x}(\mathbf{X})}\gamma(a)\prod_{b\in A_{x}^{c}(\mathbf{X})}(1-\gamma(b)), &\quad \text{ if }  A_{x}(\mathbf{X}) \not = \emptyset,  A_{x}^{c}(\mathbf{X}) \not = 
\emptyset, |S|= 1. \\
    \end{cases}
    \end{displaymath}
\end{definition}
\end{claim}

Consider $\mathbf{X}=\{x,y\}$ and $A=\{a\}$ such that $x\succ_{a}y$. This rule does not satisfy \textbf{IS} since \[p^{3}(x,\{x\},\{a\})=\gamma(a), \; p^{3}(y,\{y\},\{a\})=1-\gamma(a).\]

It is easy to verify that SCR-3 satisfies all the other Axioms.

\begin{claim} The following SCR-4 satisfies all the Axioms except \textbf{UIA}.

  Let 
\[
A_{x}^{1}(S)=\{a\in A : x=\max(\succeq_{a}, S)\}
\] and 
\[
A_{x}^{2}(S)=\{a\in A : x=\max(\succeq_{a}, S\backslash\{x^{*}_{a}(S)\})\}.
\] 
where $x^{*}_{a}(S)=\max(\succeq_{a}, S)$.
   
   \begin{definition}[SCR-4] An SCR, $p^{4}$ is SCR-4 if for all $a\in \mathbf A$ there exists an ordering $\succeq_{a}$ and attention parameters $\gamma(a)\in (\dfrac{1}{2},1]$  such that for any $S\in\mathcal{X}$, $x\in S$ and $A\in \mathcal A$,   \begin{displaymath}
p^{4}(x,S,A)=\begin{cases}
\prod_{a\in A_{x}^{1}(S)}\gamma(a)\dfrac{\prod_{b\in A_{x}^{2}(S)}(1-\gamma(b))}{(|S|-1)^{|A_{x}^{2}(S)|}}, &\quad \text{ if }  A_{x}^{1}(S) \not = \emptyset,  A_{x}^{2}(S) \not = \emptyset; \\
\ \prod_{a \in A_{x}^{1}(S)}\gamma(a), & \quad \text{if } A_{x}^{1}(S)=A; \\
 \dfrac{\prod_{b \in  A_{x}^{2}(S) }(1-\gamma(b))}{(|S|-1)^{|A_{x}^{2}(S)|}}, &\quad \text{ if } A_{x}^{2}(S)=A; \\   
 0, &\quad \text{ otherwise.}
    \end{cases}
    \end{displaymath}
\end{definition}
\end{claim}

Consider $S=\{x,y,z\}$ and $A=\{a\}$ such that $x\succ_{a}y\succ_{a}z$. This rule does not satisfy \textbf{UIA} since

\[p^{4}(x,S,A)=\gamma(a), \; p^{4}(y,S, A)=1-\gamma(a) \text{ and } p(z,S, A)=0.\]

It is easy to verify that SCR-4 satisfies all the other Axioms.

\begin{claim} The following SCR-5 satisfies all the Axioms except \textbf{DOM}.    
\end{claim}

\begin{definition}[SCR-5] An SCR $p^{5}$ is SCR-5  if for all $a\in \mathbf A$ there exists an ordering $\succeq_{a}$ and attention parameters $\gamma(a)\in (\dfrac{1}{2},1]$  such that for any $S\in\mathcal{X}$, $x\in S$ and $A\in \mathcal A$,    \begin{displaymath}
p^{5}(x,S,A)=\begin{cases} p^{s}(x,S, A), &\quad \text{ if }  |S|\geq 2; \\
 1, &\quad \text{ if } |S|=1.\\
    \end{cases}
    \end{displaymath}
\end{definition}
SCR-5 is a MAR if $|S|\geq 2$ and chooses an alternative with probability one from any singleton set. 

Consider $S=\{x,y,z\}$ and $A=\{a\}$ such that $x\succ_{a}y\succ_{a}z$. This rule does not satisfy \textbf{DOM} since 

\[ p^{5}(x,\{x,y\}, \{a\})=\gamma(a) \not = p^{5}(x,\{x\},\{a\})=1.\]

It is easy to verify that SCR-5 satisfies all the other Axioms. 

\begin{claim} The following SCR-6 satisfies all the Axioms except \textbf{BD}.    
\end{claim}

\begin{definition}[SCR-6] An SCR $p^{6}$ is SCR-6  if for all $a\in \mathbf A$ there exists an ordering $\succeq_{a}$ and attention parameters $\gamma(a)\in (\dfrac{1}{2},1]$  such that for any $S\in\mathcal{X}$, $x\in S$ and $A\in \mathcal A$,     \begin{displaymath}
p^{6}(x,S,A)=\begin{cases}
\prod_{a\in A_{x}(S)}\gamma(a)\dfrac{\prod_{b\in A_{x}^{c}(S)}(1-\gamma(b))}{(|S|-1)^{|A_{x}^{c}(S)|}}, &\quad \text{ if }  A_{x}(S) \not = \emptyset,  A_{x}^{c}(S) \not = \emptyset ;\\
\ \prod_{a \in A_{x}(S)}\gamma(a), & \quad \text{if } A_{x}(S)=A;\\
 \dfrac{\prod_{b \in  A_{x}^{c}(S) }(1-\gamma(b))}{(|S|-1)^{|A_{x}^{c}(S)|}}, &\quad \text{ if } A_{x}^{c}(S)=A.
    \end{cases}
    \end{displaymath}
\end{definition}
SCR-6 is a MAR with $\gamma(a)=\dfrac{1}{2}$ for all $a\in A$.

Consider $S=\{x,y\}$ and $A=\{a\}$ such that $x\succ_{a}y$. This rule does not satisfy \textbf{BD} since 

\[ p^{6}(x,\{x,y\}, \{a\})= p^{6}(y,\{x,y\},\{a\})=\dfrac{1}{2}.\]

It is easy to verify that SCR-6 satisfies all the other Axioms.

\begin{claim} The following example of an SCR satisfies all the Axioms except \textbf{ST} for $S=\{x,y,z\}$ and $A=\{a\}$.    
\end{claim}

\begin{proof}Consider $X=\{x,y,z\}$ and $A=\{a\}$. Let the SCR be as follows:

\[ p(x, X, A)=p(z, X, A)=\dfrac{1}{3}, \; p(y, X, A)=\dfrac{1}{3}.\]
\[p(x, \{x,y\}, A)=\dfrac{2}{3}, \; p(y, \{x,y,\}, A)=\dfrac{1}{3}.\]
\[p(y, \{y,z\}, A)=\dfrac{2}{3}, \; p(z, \{y,z\}, A)=\dfrac{1}{3}.\]
\[p(z, \{x,z\}, A)=\dfrac{2}{3}, \; p(x, \{x,z\}, A)=\dfrac{1}{3}.\]
\[ p(x, \{x\},A)=p(y,\{y\}, A)=p(z, \{z\}, A)=\dfrac{2}{3}.
\]
Note that \textbf{UIA}, \textbf{DOM} and \textbf{IN} are satisfied trivially, while \textbf{FS}, \textbf{IS} and \textbf{BD} are satisfied by the definition of the SCR above. 
\end{proof}

\subsection{Independence of the Axioms for AAR}\label{secAARindep}

We establish that the axioms characterizing the Additive Attention Rule
(Theorem \ref{thmAAR}) are independent: for each of \textbf{AVG}, \textbf{MA},
\textbf{BC}, \textbf{BD}, \textbf{IS}, and \textbf{ST} we exhibit an SCR that
satisfies the remaining five but violates that one. Since none of these rules is
an AAR (each fails a necessary axiom), each axiom is non-redundant. In every example the universe of alternatives $\mathbf{X}$ is finite, so that $\mathcal{X}$ consists of only finitely many menus and every displayed choice probability can be checked directly to lie in $[0,1]$, with menu totals not exceeding one.

\begin{definition}[AAR-1]
Let $\mathbf{X}=\{x,y\}$, $\mathbf{A}=\{a,b\}$ with $x\succ_ay$ and $y\succ_bx$, and $\gamma(a)=\gamma(b)=\tfrac12$. For every $S\in\mathcal{X}$ and $B\in\mathcal{A}$, define
\[
p^1(x,S,B):=\frac{1}{N(S)}\max_{a\in B}\gamma(a)\,V(x,a,S).
\]
\end{definition}

\textbf{Claim.} AAR-1 satisfies MA, BC, BD, IS, and ST, but violates AVG.

Since $|B|=1$ makes the maximum trivial, $p^1$ agrees with the Additive Attention Rule on every singleton attribute set, $p^1(x,S,\{a\})=\gamma(a)V(x,a,S)/N(S)$; consequently MA, BC, BD, and ST --- each of which constrains only single-attribute choice probabilities --- hold exactly as in the proof of Theorem \ref{thmAAR}, and IS holds since $p^1(x,\{x\},\{a\})=\gamma(a)$ for every $x$. It violates AVG: with $N(\{x,y\})=3$,
\[
p^1(x,\{x,y\},\{a\})=\tfrac13\Bigl(\tfrac12\Bigr)(2)=\tfrac13,\qquad p^1(x,\{x,y\},\{b\})=\tfrac13\Bigl(\tfrac12\Bigr)(1)=\tfrac16,
\]
so AVG would require $p^1(x,\{x,y\},\{a,b\})=\tfrac12\bigl(\tfrac13+\tfrac16\bigr)=\tfrac14$, whereas
\[
p^1(x,\{x,y\},\{a,b\})=\tfrac13\max\Bigl\{\tfrac12(2),\tfrac12(1)\Bigr\}=\tfrac13(1)=\tfrac13\ne\tfrac14.
\]

\begin{definition}[AAR-2]
Let $\mathbf{X}=\{x,y,z\}$, $\mathbf A=\{a\}$, $x\succ_ay\succ_az$, and $\gamma(a)=\tfrac12$. On singleton and pair menus let $p^2$ coincide with the Additive Attention Rule: $p^2(w,\{w\},a)=\tfrac12$ for $w\in\{x,y,z\}$, and on each pair the $\succeq_a$-winner receives $\tfrac13$ and the loser $\tfrac16$. On the triple set define $p^2(x,S,a)=p^2(y,S,a)=p^2(z,S,a)=\tfrac16$.
\end{definition}

\textbf{Claim.} AAR-2 satisfies AVG, BC, BD, IS, and ST, but violates MA.

Since the singleton and pair values are those of the Additive Attention Rule, IS, BD, BC, and ST hold exactly as in the proof of Theorem \ref{thmAAR}, and AVG holds trivially since $|A|=1$. It violates MA: taking $S=\{x,y\}$, $T=\{x,z\}$, so $S\cap T=\{x\}$,
\[
N(S)p^2(x,S,a)+N(T)p^2(x,T,a)-p^2(x,\{x\},a)=3\Bigl(\tfrac13\Bigr)+3\Bigl(\tfrac13\Bigr)-\tfrac12=\tfrac32,
\]
so MA would require $p^2(x,\{x,y,z\},a)=\tfrac{3/2}{N(\{x,y,z\})}=\tfrac{3/2}{6}=\tfrac14$, but the assigned value is $\tfrac16\ne\tfrac14$.

\begin{definition}[AAR-3]
Let $\mathbf{X}=\{x,y,z\}$, $\mathbf A=\{a\}$, $x\succ_ay\succ_az$, and $\gamma(a)=\tfrac14$. With
\[
W(x,a,S):=1+2\bigl|\{y\in S\setminus\{x\}: x\succ_ay\}\bigr|,
\]
define $p^3(x,S,\{a\}):=\dfrac{\gamma(a)}{N(S)}W(x,a,S)$.
\end{definition}

\textbf{Claim.} AAR-3 satisfies AVG, MA, BD, IS, and ST, but violates BC.

AVG holds trivially since $|A|=1$. Since $S\setminus\{x\}$ and $T\setminus\{x\}$ are disjoint whenever $S\cap T=\{x\}$, the win-count in $W$ is additive across such a union exactly as $V$ is, so $W(x,a,S\cup T)=W(x,a,S)+W(x,a,T)-1$; the telescoping argument used to prove MA in Theorem \ref{thmAAR} therefore applies verbatim with $W$ in place of $V$, and MA holds. On singletons $W(w,a,\{w\})=1$, so $p^3(w,\{w\},a)=\gamma(a)=\tfrac14$ for every $w$, giving IS. On a pair, the winner has one win ($W=3$) and the loser none ($W=1$), so with $N=3$,
\[
p^3(\text{winner},\cdot,a)=\tfrac14\Bigl(\tfrac33\Bigr)=\tfrac14,\qquad p^3(\text{loser},\cdot,a)=\tfrac14\Bigl(\tfrac13\Bigr)=\tfrac1{12},
\]
which are distinct (BD), and since the winner of every pair is always the $\succeq_a$-larger alternative, the binary order recovered from $p^3$ coincides with $\succeq_a$ itself, giving ST. It violates BC: BC requires the winner to receive $\tfrac23$ of its own singleton value, i.e.\ $\tfrac23\bigl(\tfrac14\bigr)=\tfrac16$, but the winner actually receives $\tfrac14\ne\tfrac16$.

\begin{definition}[AAR-4]
Let $\mathbf{X}=\{x,y\}$, $\mathbf A=\{a_0,a_1\}$ with $\gamma(a_0)=0$, $\gamma(a_1)=\tfrac12$, and any linear orders on $\mathbf{X}$ under $a_0$ and $a_1$ (say $x\succ_{a_1}y$). Let $p^4$ be given by the Additive Attention Rule expression itself, evaluated at these weights:
\[
p^4(x,S,B):=\frac{1}{N(S)|B|}\sum_{a\in B}\gamma(a)\,V(x,a,S).
\]
\end{definition}

\textbf{Claim.} AAR-4 satisfies AVG, MA, BC, IS, and ST, but violates BD.

AVG holds by construction. The derivations of MA and BC in the proof of Theorem \ref{thmAAR} use only that $p(x,S,\{a\})=\gamma(a)V(x,a,S)/N(S)$ for some $\gamma(a)\ge0$ --- nowhere do they require $\gamma(a)>0$ --- so MA and BC hold for every $a\in A$, including $a_0$. IS holds since $p^4(x,\{x\},\{a_0\})=0=p^4(y,\{y\},\{a_0\})$. ST holds for $a_1$ as usual, and holds vacuously for $a_0$, since its hypothesis (a strict binary inequality under $a_0$) is never satisfied: $p^4(x,\{x,y\},\{a_0\})=p^4(y,\{x,y\},\{a_0\})=0$. This is exactly the violation of BD: the two binary values under $a_0$ coincide rather than differ. Since the Additive Attention Rule requires $\gamma(a)>0$ for every attribute, AAR-4 is not an instance of it; this example shows that BD is precisely what forces the attention weights to be strictly positive.

\begin{definition}[AAR-5]
Let $\mathbf{X}=\{x,y,z\}$, $\mathbf A=\{a\}$, $x\succ_ay\succ_az$, and $\sigma(x)=\tfrac3{10}$, $\sigma(y)=\tfrac14$, $\sigma(z)=\tfrac15$. Define
\[
p^5(w,S,\{a\}):=\frac{\sigma(w)}{N(S)}V(w,a,S)\qquad\text{for }w\in S.
\]
\end{definition}

\textbf{Claim.} AAR-5 satisfies AVG, MA, BC, BD, and ST, but violates IS.

AVG holds trivially since $|A|=1$. Because $V(w,a,S\cup T)=V(w,a,S)+V(w,a,T)-1$ whenever $S\cap T=\{w\}$ --- an identity about $w$'s own rank alone --- the telescoping argument of Theorem \ref{thmAAR} applies with the base value $\sigma(w)$ in place of a common $\gamma(a)$, so $N(S)p^5(w,S,a)=\sigma(w)V(w,a,S)$ satisfies MA. On a pair $\{w,w'\}$ with $w\succ_aw'$, $N=3$ and $(V(w,\cdot),V(w',\cdot))=(2,1)$, so
\[
p^5(w,\{w,w'\},a)=\tfrac23\sigma(w),\qquad p^5(w',\{w,w'\},a)=\tfrac13\sigma(w'),
\]
which is exactly BC's requirement, stated relative to each alternative's own singleton value, so BC holds regardless of whether $\sigma$ is constant. Numerically,
\[
p^5(x,\{x,y\},a)=\tfrac15,\quad p^5(y,\{x,y\},a)=\tfrac1{12},\quad p^5(x,\{x,z\},a)=\tfrac15,\quad p^5(z,\{x,z\},a)=\tfrac1{15},
\]
\[
p^5(y,\{y,z\},a)=\tfrac16,\quad p^5(z,\{y,z\},a)=\tfrac1{15},
\]
which are distinct within each pair (BD). Since $\sigma(x),\sigma(y),\sigma(z)$ lie within a band of ratio below $2$ of one another, the alternative ranked higher under $\succeq_a$ always receives the higher binary probability, so the order recovered from binary comparisons coincides with $\succeq_a$ itself, giving ST. It violates IS: $p^5(x,\{x\},a)=\tfrac3{10}\ne\tfrac14=p^5(y,\{y\},a)$.

\begin{definition}[AAR-6]
Let $\mathbf{X}=\{x,y,z\}$, $\mathbf A=\{a\}$, $\gamma(a)=\tfrac12$, and replace the linear order $\succeq_a$ by the cyclic relation $x\text{ beats }y$, $y\text{ beats }z$, $z\text{ beats }x$ (not transitive). On singletons let $p^6(w,\{w\},a)=\tfrac12$ for every $w$. On each pair let the designated winner receive $\tfrac23\bigl(\tfrac12\bigr)=\tfrac13$ and the loser receive $\tfrac13\bigl(\tfrac12\bigr)=\tfrac16$. On the triple set let $p^6(x,S,a)=p^6(y,S,a)=p^6(z,S,a)=\tfrac16$.
\end{definition}

\textbf{Claim.} AAR-6 satisfies AVG, BC, BD, IS, and MA, but violates ST.

AVG holds trivially since $|A|=1$. IS, BD, and BC hold by direct inspection: singleton values are equal, pair values are distinct, and each pair matches the $\tfrac23$-$\tfrac13$ rule relative to its own singleton value. MA holds: taking $S=\{x,y\}$, $T=\{x,z\}$ (so $S\cap T=\{x\}$),
\[
N(S)p^6(x,S,a)+N(T)p^6(x,T,a)-p^6(x,\{x\},a)=3\Bigl(\tfrac13\Bigr)+3\Bigl(\tfrac16\Bigr)-\tfrac12=1,
\]
matching $N(S\cup T)p^6(x,S\cup T,a)=6\bigl(\tfrac16\bigr)=1$; by the symmetry of the cycle, the same check holds starting from $y$ or from $z$. It violates ST: $x$ beats $y$ ($\tfrac13>\tfrac16$) and $y$ beats $z$ ($\tfrac13>\tfrac16$), yet on $\{x,z\}$ it is $z$ --- not $x$ --- that wins, since $p^6(z,\{x,z\},a)=\tfrac13>\tfrac16=p^6(x,\{x,z\},a)$. The cyclic relation admits no linear order, so AAR-6 is not an Additive Attention Rule either.

\end{document}